\documentclass[12pt]{article}
\usepackage{amssymb, amsmath}    
\usepackage{amsthm}
\usepackage{url}
\usepackage[titletoc,toc,title]{appendix}
\theoremstyle{definition}
\newtheorem{lemma}{Lemma}[section]

\newtheorem{thm}{Theorem}[section]
\newtheorem{remark}{Remark}[section]

\newcommand{\norm}[1]{\left\lVert#1\right\rVert}
\newcommand{\abs}[1]{\left|#1\right|}
\newcommand{\unb}{\underbrace}

\usepackage{subfig,tikz}
\usetikzlibrary{decorations.pathmorphing}
\usepackage{graphics}

\newcommand{\td}{\text{d}}
\def\be{\begin{equation}}
\def\ee{\end{equation}}
\def\bea{\begin{eqnarray}}
\def\eea{\end{eqnarray}}

\title{\bf{Small deformations of extreme five dimensional Myers-Perry  black hole initial data}}

\author{Aghil Alaee\footnote{aak818@mun.ca } \,\,and  Hari K. Kunduri\footnote{hkkunduri@mun.ca } \\ \\
\small \sl $^a$ Department of Mathematics and Statistics, \\  \small \sl Memorial University of Newfoundland \\ \small \sl St John's NL A1C 4P5, Canada}

\date{}

\begin{document}

\maketitle

%\vskip1.5cm

%\begin{abstract}

%\end{abstract}

%\newpage

%\tableofcontents
%\newpage

\begin{abstract}
We demonstrate the existence of a one-parameter family of initial data for the vacuum Einstein equations in five dimensions representing small deformations of the extreme Myers-Perry black hole. This initial data set has `$t-\phi^i$' symmetry and preserves the angular momenta and horizon geometry of the extreme solution. Our proof is based upon an earlier result of Dain and Gabach-Clement concerning the existence of $U(1)$-invariant initial data sets which preserve the geometry of extreme Kerr (at least for short times).  In addition, we construct a general class of transverse, traceless symmetric rank 2 tensors in these geometries.
\end{abstract}

\section{Introduction}
\label{intro}
Einstein's equations admit an initial-value formulation, with Cauchy data specified by the triple $(\Sigma, h,K)$ where $\Sigma$ is a Riemannian manifold  equipped with a metric tensor $h$ and $K$ represents the second fundamental form of $\Sigma$,  regarded as a spacelike hypersurface of spacetime.  The field equations, together with the Gauss-Codazzi equations,  impose the constraints
\begin{eqnarray}\label{eq:constraints}
R_h - K^{ab}K_{ab} + (\text{tr} K)^2 &=& 8\pi \mu \nonumber \\ \nabla^b\left(K_{ab} - \text{tr} K h_{ab}\right) &=& -4\pi j_a
\end{eqnarray} where $R_h$ is the scalar curvature of $(\Sigma,h)$ and $(\mu,j)$ are the local energy density and momentum current respectively. The complete proof  that solutions to the constraints evolve into a unique maximal development, for sufficient regularity, is a significant achievement (see e.g. \cite{choquet2009general} and for a concise summary \cite{bartnik2004constraint}) . It is an important problem to actually construct initial data with desired properties. This involves identifying the freely specifiable `degrees of freedom' and then determining whether a corresponding solution exists and is unique. 

A useful approach to achieve this is the conformal method (\cite{choquet1980cauchy,choquet2000einstein,Dain:2010uh}). In the special case of data with constant mean curvature ($\text{tr} K =$ const) the problem reduces to solving a conformally invariant system of equations for the conformal factor and a vector field which generates the extrinsic curvature. For spatially closed and asymptotically Euclidean initial data sets, one can prove existence using the conformal method \cite{choquet2000einstein} (for spacetime dimension $D \geq 4$). Subsequently,  Maxwell \cite{maxwell2005solutions} constructed asymptotically Euclidean initial data with apparent horizon boundary conditions (in particular, he treated the case with multiple disconnected apparent horizons). This case is naturally relevant to black holes. 

While the above results are powerful in their generality, one can also consider the existence of initial data with very specific geometrical properties.  This paper will be concerned with initial data sets which have one Euclidean end and one cylindrical end. Roughly, the latter means an initial data set $(\Sigma,h)$ has an asymptotic end which is  diffeomorphic to $\mathbb{R} \times N$ where $N$ is a compact manifold.  A systematic analysis of initial data on manifolds with cylindrical ends  was performed in \cite{chrusciel4937initial,chrusciel5138initial}.  In particular, existence of solutions of Lichnerowiscz's equation is proved using the powerful barrier method \cite{isenberg1995constant}. The purpose of our analysis, however, is to prove the existence of a rather specific class of perturbed initial data with additional properties (e.g. preserving angular momenta of the background data).  We will make clear at the end of this section how our results are related to  \cite{chrusciel4937initial,chrusciel5138initial}.

Initial data sets with cylindrical ends arise within the context of stationary, extreme black holes. Extreme black holes with degenerate Killing horizons have vanishing surface gravity $\kappa =0$, and in the limit as one approaches the horizon, Einstein's equations  decouple in a precise manner into a set of equations defined only on the horizon \cite{Kunduri:2013gce}. This gives rise to the notion of a \emph{near-horizon geometry}, which if often thought of as an infinite `throat' region in the spacetime (indeed the proper length to a spatial section of the horizon is infinite).   

 Extreme black holes have attracted a great deal of interest in recent years. Due to the decoupling described above, classifying near-horizon geometries is tractable and yields important information on the full space of extreme solutions (e.g. allowed geometries and topologies of spatial cross sections). Furthermore, extreme black hole geometries saturate a number of geometric inequalities which must hold for initial data sets and for marginally outer trapped surfaces in four dimensions  \cite{dain2006proof,dain2008proof,chrusciel2009mass} (see also \cite{hollands2012horizon} for work on the latter problem in $D>4$).   Finally, extreme black holes have the simplest microscopic description within string theory, and so are an important testing ground for various calculations in quantum gravity, the most well-known of which is black hole entropy counting.  Recently, due to the work of Aretakis  and others \cite{aretakis2011stability,Aretakis:2012ei,lucietti2012gravitational,Lucietti:2012xr,Murata:2013daa}, extreme black holes have been shown to be unstable to a certain horizon instability. An alternative approach to studying the non-linear instability of the extreme Kerr-Newman family using perturbations of the initial data of extreme Reisnner-Nordstsrom also has recently appeared \cite{reiris2013instability}.

A spacelike slice of such a near-horizon geometry has the form of the geometry of a cylindrical end, where $N \cong H$, a spatial cross-section of the horizon.  Hence initial data for an asymptotically flat extreme black hole has one asymptotic Euclidean region and an asymptotically cylindrical end.  The simplest example of this occurs for initial data of the extreme $M=\sqrt{J}$ Kerr black hole \cite{Dain:2010uh}. These authors, using the conformal method alluded to above, proved that there exists a one-parameter family of axisymmetric initial data of the vacuum Einstein equations which preserve the asymptotic behaviour, angular momenta, and area of the cylindrical end (this area corresponds to the area of the spatial sections of the horizon of the Kerr black hole). In particular, as a consequence of the geometric inequalities, one can show the energy of  any member of this family must be strictly greater than that of the extreme Kerr initial data.  Note that the solutions satisfy weak regularity conditions (i.e. they belong to a certain Sobolev space) and in particular are not generically smooth, let alone analytic. This last distinction could be important when considering the evolution of this initial data. The extreme Kerr black hole is known to be the unique  (analytic) vacuum, stationary, rotating asymptotically flat spacetime containing a single degenerate horizon \cite{Hollands:2008wn}\cite{Amsel:2009et,Figueras:2009ci} . Hence the evolution of the initial data sets discussed above could settle down to non-analytic asymptotically flat (possibly stationary) extreme black holes. Of course, we cannot address this issue without understanding the evolution.  

It is natural to investigate the possibility of extending the result of \cite{Dain:2010uh} to extreme, five-dimensional black holes. The simplest candidate would be extreme Myers-Perry black hole \cite{myers1986black}, which is qualitatively similar to Kerr. A maximal slice can be found with $U(1)^2$ isometry and has topology $\mathbb{R} \times S^3$ \cite{0264-9381-31-5-055004}. However there are two main differences as one moves from $n=3$ to $n=4$ spatial dimensions. First, it turns out we will have to construct solutions of the constraint equations which belong to  Bartnik's weighted Sobolev spaces $W'^{k,p}_\delta$ \cite{bartnik1986mass}. Our asymptotic fall-off conditions at the Euclidean end and cylindrical end require $kp > n$ (see Lemma A.1 in \cite{Dain:2010uh}). We only require weak differentiability to second order, so we take $(k,p,\delta) = (2,3,-1)$\footnote{One could also take $(k,p,\delta) = (3,2,-1)$ but this leads to a stronger regularity condition for a particular elliptic operator %$\mathcal{E}(\lambda,u)$in target space%
and the functions in the background metric do not satisfy this regularity.} whereas in the analysis of \cite{Dain:2010uh}, $(k,p,\delta)=(2,2,-1/2)$. The latter spaces are weighted Hilbert spaces, which prove extremely useful in the elegant construction given in \cite{Dain:2010uh}.  Second, we require five scalar functions to characterize our data as opposed to two and our geometries have $U(1)^2$ symmetry which complicates the parameterization of the extrinsic curvature.  \par \noindent Our main result is Theorem \ref{maintheorem} and it can be informally stated as follows:\\
\begin{itshape}
\par \noindent There exists a one parameter, $U(1)^2$-invariant, maximal family  of solutions to Einstein's constraint equations. This family of data is second order differentiable with respect to an appropriate norm and it has the same angular momentum and area of the event horizon of an extreme Myers-Perry black hole. Moreover, the geometry of this family is close (in a suitable sense) to the extreme Myers-Perry initial data set.
\end{itshape}  \\

It is important to clarify what is new about this result and how it is related to the analysis of \cite{chrusciel4937initial,chrusciel5138initial}. In particular, Theorem 6.1 of \cite{chrusciel4937initial} asserts the existence of a class of solutions to Lichnerowicz's equation for complete initial data with non-negative scalar curvature and strictly positive scalar curvature on cylindrical ends.   These results are quite powerful and general in that no symmetry assumptions are made on the data.  However, if one wishes to impose additional conditions (e.g. axisymmetry) on the data, one might be interested if there exists special families of data with the same ADM energy, conserved angular momenta and/ or area of the cylindrical end.  This work is concerned with finding a class of initial data suitably close to the extreme Myers-Perry data that preserves the angular momenta and area of its cylindrical end. This data can be interpreted as perturbations of extreme Myers-Perry.  To prove this result, we need to first consider a more general problem of finding transverse, traceless symmetric rank 2 tensors on $U(1)^{2}$-invariant geometries which generalizes \cite{dain2001initial}.  To the best of our knowledge, this work has not appeared before and should be useful in various contexts when considering initial data with symmetries. 

\par This paper is organized as follows. In Section 2 we discuss the maximal slices of the extreme Myers-Perry solution.  Section 3 states our theorem and Section 4 provides most of the technical details in the proof. We conclude with a discussion. The appendices collects a number of useful theorems which we use in our proof, and some technical properties of the Myers-Perry solution which we use to establish our result. 
\section{Initial Data with $U(1)^2$ symmetry}
In this work we consider general initial data sets $(\Sigma,h,K)$ which are invariant under $U(1)^2$ isometry.  In addition we will restrict attention to maximal slices, i.e. $\text{tr} K =0$.  Finally, as explained in detail below, the Myers-Perry maximal initial data set of interest has a further useful property (`$t-\phi^i$' symmetry) and we will impose this on our class of initial data sets as well. In the following for convenience we will simply refer to initial data satisfying these various conditions as `biaxisymmetric'.
\subsection{Extreme Myers-Perry black hole and initial data}\label{sec1}
Our starting point is the five-dimensional vacuum Myers-Perry black hole $(M,g)$  with metric \cite{myers2011myers}
\begin{eqnarray}
g&=&-\td t^2+\frac{\mu}{\Sigma}\left(\td t+a\sin^2\theta \td\varphi+b\cos^2\theta \td\psi\right)^2+\frac{\tilde{r}^2\Sigma}{\Delta({\tilde{r}})} \td  \tilde{r}^2
+ \Sigma \td\theta^2 \nonumber \\ &+& \left(\tilde{r}^2+a^2\right)\sin^2\theta \td\varphi^2
+\left(\tilde{r}^2+b^2\right)\cos^2\theta \td\psi^2
\end{eqnarray}
where
\begin{gather}
\Sigma=\tilde{r}^2+b^2\sin^2\theta+a^2\cos^2\theta,\\
\Delta({\tilde{r}})=\left(\tilde{r}^2+a^2\right)\left(\tilde{r}^2+b^2\right)-\mu \tilde{r}^2.
\end{gather}
The solution is parameterized by $(\mu,a,b)$ with orthogonally transitive isometry group $\mathbb{R}_t\times U(1)^2$, where $\mathbb{R}$ is the time translation symmetry and $U(1)^2$ is the rotational symmetry generated by $\partial_{\psi}$ and $\partial_{\varphi}$. Here $(\tilde{r},\theta)$ parameterize the two-dimensional surfaces orthogonal to orbits of the isometry group.   We take $\mu >0$ so that the mass of the spacetime $M>0$ and without loss of generality we take $a,b >0$. The horizons of this black hole are located at the roots of $\Delta({\tilde{r}})$, denoted $\tilde{r}_{H\pm}$. %=\pm\sqrt{\frac{\mu-a^2-b^2+\sqrt{\left(\mu-a^2-b^2\right)-4a^2b^2}}{2}}$.
The metric is written in a chart that covers the black hole exterior $\tilde{r}_{H+} < \tilde{r} < \infty$. In addition $0 < \theta < \pi/2$, and $\psi,\phi$ are periodic with period $2\pi$. 
%Moreover, the singularity of this metric for nonvanishing $a$ and $b$ with $a^2\neq b^2$ is located at roots of $\Sigma$, i.e. 
%\begin{equation}
%\sin^2\theta=\frac{\tilde{r}^2-a^2}{b^2-a^2}
%\end{equation}
As is well known, the solution is qualitatively similar to the Kerr solution. In the extreme limit, $\mu=(a+b)^2$ and $\Delta({\tilde{r}})=\left(\tilde{r}^2-ab\right)^2$. We define a new radial coordinate $r^2=\tilde{r}^2-ab$, resulting in
\begin{eqnarray}
g&=&-\td t^2+\frac{\mu}{\Sigma}\left(\td t+a\sin^2\theta \td\varphi+b\cos^2\theta \td\psi\right)^2+\frac{\Sigma}{r^2}\td r^2
+ \Sigma \td \theta^2 \nonumber\\ &+&\left(r^2+ab+a^2\right)\sin^2\theta \td\varphi^2
+\left(r^2+ab+b^2\right)\cos^2\theta \td\psi^2.
\end{eqnarray}
%*********************figure***********************************************************
\begin{figure}
\centering
\subfloat{\label{fig:1(b)}
\begin{tikzpicture}[scale=1.1, every node/.style={scale=0.6}]
\fill[gray!40!white](0,0)--(1,1)--(2,0)--(1,-1)--(0,0);
\draw[black,thick](0,0)--(-1,1)--(0,2)--(1,1);
\draw[black,thick](0,0)--(1,-1)--(2,0)--(1,1);
\draw[black,thick](0,0)--(-1,-1)--(1,-3)--(2,-2)--(1,-1);
\draw[black,thick](1,1)--(2,2);
\draw[black,thick](0,0)--(1,1);
%\draw[black,thick](0,2)--(1,3);
\draw[black,thick](0,-2)--(1,-1);
%*****************************************************
%\draw[black!50!green,dashed](1,-1)..controls (1.6,-.4) and (1.6,.4)..(1,1);
%\draw[black!50!green,dashed](1,-1)..controls (1.4,-.4) and (1.4,.4)..(1,1);
%\draw[black!50!green,dashed](1,-1)..controls (1.2,-.4) and (1.2,.4)..(1,1);
%\draw[black!50!green,dashed](1,-1)--(1,1);
%\draw[black!50!green,dashed](1,-1)..controls (.6,-.4) and (.6,.4)..(1,1);
%\draw[black!50!green,dashed](1,-1)..controls (.4,-.4) and (.4,.4)..(1,1);
%\draw[black!50!green,dashed](1,-1)..controls (.2,-.4) and (.2,.4)..(1,1);
%*************************************
%\draw[blue!50!white,dashed](0,0)..controls (.6,.7) and (.7,1.3)..(0,2);
%\draw[blue!50!white,dashed](0,0)..controls (.4,.7) and (.5,1.3)..(0,2);
%\draw[blue!50!white,dashed](0,0)..controls (.2,.7) and (.3,1.3)..(0,2);
%\draw[blue!50!white,dashed](0,0)..controls (.7,.6) and (1.3,.7)..(2,0);
%\draw[blue!50!white,dashed](0,0)..controls (.7,.4) and (1.3,.5)..(2,0);
%\draw[blue!50!white,dashed](0,0)..controls (.7,.2) and (1.3,.3)..(2,0);
%\draw[blue!50!white,dashed](0,0)--(2,0);
%\draw[blue!50!white,dashed](0,0)..controls (.7,-.6) and (1.3,-.7)..(2,0);
%\draw[blue!50!white,dashed](0,0)..controls (.7,-.4) and (1.3,-.5)..(2,0);
%\draw[blue!50!white,dashed](0,0)..controls (.7,-.2) and (1.3,-.3)..(2,0);
%*********************************
\draw[red,thick,decorate, decoration={zigzag,segment length = .8mm, amplitude = .2mm}](0,-2)..controls(.2,-1.2)and(.2,-.8)..(0,0);
\draw[red,thick,decorate, decoration={zigzag,segment length = .8mm, amplitude = .2mm}](0,0)..controls(.2,.8)and(.2,1.2)..(0,2)node[black,rotate=90,left=16mm,above=.5mm]{Singularity};
\draw[red,thick,decorate, decoration={zigzag,segment length = .8mm, amplitude = .2mm}](0,-2)..controls(.1,-2.4)and(.1,-2.8)..(0,-3);
\draw[blue,thick]node[black,thick,right=2cm,above=.3mm,font=\large]{$\Sigma$}(0,0)--(2,0);
\draw[gray,->]node[black,right=155pt,above=20pt]{Event horizon}(2.1,.6)--(.65,.6);
\draw[->] (3.2,-.5)..controls(2.8,0) and (2.5,.3)..(1.5,.1)node[black,right=3.1cm,below=1.2cm]{Slice};
\end{tikzpicture}}
\caption{Carter-Penrose diagram of extreme Myers-Perry black hole. The gray region is domain of outer communication=DOC}
\end{figure}
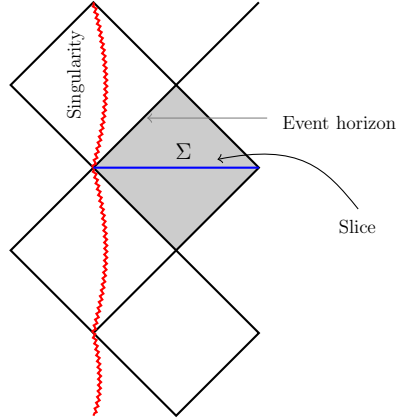
%*********************figure*********************************************************** 

where $r>0$. There is a degenerate Killing horizon located at $r=0$, which can be seen by transforming to an adapted Gaussian coordinate system.  The  mass and angular momenta are easily evaluated using Komar integrals, giving
\begin{equation}
M=\frac{3\pi}{8}\mu,\quad\quad\quad J^{\varphi}=\frac{\pi \mu a}{4},\quad\quad\quad J^{\psi}=\frac{\pi \mu b}{4} \, .
\end{equation} The extremality  condition can be written 
\begin{equation}
M^3 = \frac{27 \pi}{32} (J^\psi + J^\varphi)^2
\end{equation}
Consider a spacelike hypersurface $\Sigma$ corresponding to a  $t =$ constant slice in the above geometry.  The induced metric and extrinsic curvature $\Sigma$ is easily found to be 
\begin{eqnarray}
h&=&\frac{\Sigma}{r^2}\td r^2
+ \Sigma \td \theta^2+\left(\frac{a^2\sin^2\theta\mu}{\Sigma}+\left(r^2+ab+a^2\right)\right)\sin^2\theta \td \varphi^2+2\frac{ab\cos^2\theta\sin^2\theta\mu}{\Sigma} \td\varphi \td\psi\nonumber\\
&+&\left(\frac{b^2\cos^2\theta\mu}{\Sigma}+\left(r^2+ab+b^2\right)\right)\cos^2\theta \td\psi^2 \label{slicemet}
\end{eqnarray}
\begin{eqnarray}
K&=&-\frac{a\mu \left(r^2+ab+b^2\right)\left(\Sigma+r^2+ab+a^2\right)}{\Sigma^2\sqrt{g^{tt}}r^3}\sin^2\theta drd\varphi+\frac{a\mu(a^2-b^2)\cos\theta\sin^3\theta}{\Sigma^2\sqrt{g^{tt}}} d\theta d\varphi\nonumber\\
&-&\frac{b\mu \left(r^2+ab+a^2\right)\left(\Sigma+r^2+ab+b^2\right)}{\Sigma^2\sqrt{g^{tt}}r^3}\cos^2\theta drd\psi+\frac{b\mu(a^2-b^2)\cos^3\theta\sin\theta}{\Sigma^2\sqrt{g^{tt}}} d\theta d\psi\nonumber\\
\end{eqnarray} Although not time-symmetric,  this initial data has in addition '$t-\phi^i$' symmetry (under the simultaneous diffeomorphisms $(\varphi, \psi) \to (-\varphi, -\psi)$ $h$ is invariant and $K$ reverses sign) \cite{figueras2011black}. This symmetry in particular implies $\text{tr} K =0$, i.e. the slices are maximal. The triple $(\Sigma, h,K)$ forms a vacuum maximal initial data set (i.e. a solution of \eqref{eq:constraints} with $\mu = j= 0$) for the extreme black hole exterior:
\begin{equation}\label{eq:constraints2}
R_h - K^{ab}K_{ab} =0 \qquad \nabla^bK_{ab} = 0
\end{equation} 
The pair $(\Sigma,h)$ represents a Riemannian manifold with one asymptotically flat end and one asymptotically cylindrical end. $\Sigma$ is diffeormorphic to $\mathbb{R}\times S^3 \cong \mathbb{R}^4\backslash \{0\}$ \cite{0264-9381-31-5-055004} and the spatial metric \eqref{slicemet}  is  a cohomogeneity two,  asymptotically flat metric that extends globally onto $\Sigma$.  The metric has the following fall off conditions at its asymptotically flat end:
\begin{gather}
h_{ab}=\delta_{ab}+\mathcal{O}(r^{-2}),\quad\quad\quad \partial h_{ab}=\mathcal{O}(r^{-3})\\
K_{ab}=\mathcal{O}(r^{-4}),\quad\quad\quad \partial K_{ab}=\mathcal{O}(r^{-5})
\end{gather}
where $\delta$ is the Euclidean metric on $\mathbb{R}^4$.
To investigate the geometry of \eqref{slicemet} as $r \to 0$, perform the transformation $s = -\ln r$. This reveals a new asymptotic region corresponding to the limit $s \to \infty$ with  geometry 
\begin{equation}
h= \sigma(\theta) \left[\td s^2 + \td\theta^2 + \frac{(a+b)^2}{\sigma(\theta)^2}\left(a^2 \sin^2\theta \td \varphi^2 + b^2\cos^2\theta \td \psi^2+ ab\left[\cos^2\theta \td \varphi + \sin^2\theta \td \psi\right]^2\right)\right]
\end{equation} 
\begin{eqnarray}
K=\frac{\mu\sqrt{ab}(a\cos^2\theta+b\sin^2\theta+a)}{\sqrt{(a\cos^2\theta+b\sin^2\theta)^3(a+b)^3}}\sin^2\theta dsd\varphi+\frac{\mu\sqrt{ab}(a\cos^2\theta+b\sin^2\theta+b)}{\sqrt{(a\cos^2\theta+b\sin^2\theta)^3(a+b)^3}}\cos^2\theta dsd\psi\nonumber
\end{eqnarray}It can be shown that the $(\theta,\varphi,\psi)$ part of the metric can be globally extended to an inhomogeneous metric on $S^3$ \cite{kunduri2009classification}. Thus as $r \to 0$ $(\Sigma,h)$ has a cylindrical end with geometry $\mathbb{R} \times S^3$.  A schematic diagram of the slice $(\Sigma,h)$ is given in Figure \ref{fig2}(a). 
%\paragraph{we do not need the following now}
%It is interesting to point out that we can do coordinate transformation $s=-\ln r$ and write the slice metric for cylindrical end, equation (\ref{cylindrical}). Hence, $s\rightarrow\infty$ corresponds to $r\rightarrow 0$. Also the extrinsic curvature $K_{ab}$ is obtain from normal vector to $\Sigma$ in $(M,g)$, i.e $K_{ab}=h^{c}_{b}\bar{\nabla}_{c}n_{a}$ and $n_a=\frac{dt}{\sqrt{g^{tt}}}$ where $\bar{\nabla}$ is connection respect to $g$. So the triple $(\Sigma,h_{ab},K_{ab})$ is initial data of extreme Myers-Perry black which satisfy in Einstein's constraint equations in vacuum
%\begin{eqnarray}\label{eq:constraints2}
%R_h - K^{ab}K_{ab} + (\tr K)^2 &=&0 \nonumber \\ 
%\nabla^b\left(K_{ab} - \tr K h_{ab}\right) &=& 0
%\end{eqnarray}
%where $\nabla$ and $R_h$ are connection and Ricci scalar curvature of $h$ and $\tr K=h^{ab}K_{ab}$. The Rimeannian manifold $(\Sigma,h_{ab})$ similar to extreme Kerr slice manifold has one asymptotically flat end and one cylindrical end. The maximal analytic extension of initial data set $(\Sigma,h_{ab},K_{ab})$ is shown in gray region of figure \ref{fig:1(b)}.

%***************************************figure******************************************
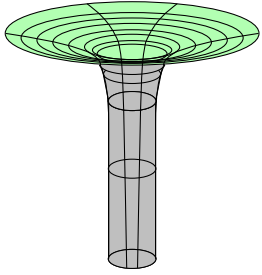
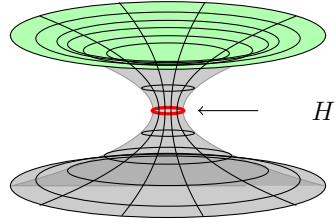
\begin{figure}
\centering
\subfloat[{Spatial slice of an extreme black hole with one cylindrical end and one asymptotic end.}]{\label{fig:2(a)}
\begin{tikzpicture}[scale=.6, every node/.style={scale=0.7}]
\draw[fill=gray!50!white](-.52,-5)--(0,-5)--(0,-5) ellipse (15pt and 6pt)--(0,-5)--(.52,-5)--(.52,-1.5)--(0,-1.5)--(0,-1.5) ellipse (15pt and 6pt)--(0,-1.5)--(-.52,-1.5)--(-.52,-5)--(0,-5);
\draw[fill=gray!50!white](-.52,-1.5)--(0,-1.5)--(0,-1.5) ellipse (15pt and 6pt)--(0,-1.5)--(.52,-1.5)--(.68,-.7)--(0,-.7)--(0,-.7) ellipse (18pt and 6pt)--(0,-.7)--(-.68,-.7)--(-.52,-1.5)--(0,-1.5);
\draw[fill=gray!50!white](0,-.9) ellipse (18pt and 6pt);
\draw[fill=gray!50!white] (0,-.7) ellipse (20pt and 6pt);
\draw[fill=gray!50!white] (0,-1.5) ellipse (15pt and 6pt);
\draw[fill=gray!50!white] (0,-3) ellipse (15pt and 6pt);
\draw[fill=gray!50!white] (0,-5) ellipse (15pt and 6pt);
\draw[fill=green!30!white] (0,0) ellipse (80pt and 20pt);
\draw[black] (0,-.1) ellipse (70pt and 16pt);
\draw[black] (0,-.15) ellipse (60pt and 14pt);
\draw[black] (0,-.2) ellipse (50pt and 12pt);
\draw[black] (0,-.27) ellipse (40pt and 10pt);
\draw[black] (0,-.3) ellipse (30pt and 8pt);
\draw[black] (0,-.37) ellipse (25pt and 6pt);
\draw[black] (0,-.45) ellipse (22pt and 6pt);
\draw (2.8,0).. controls (1.5,-.2)and(1.3,-.3) .. (1,-.4)..controls (.8,-.5)and(.7,-.6) .. (.7,-.7)..controls (.8,-.7)and(.52,-.7) .. (.52,-2)..controls (.52,-2.5)and(.52,-3) .. (.52,-5);
\draw (1,.66).. controls (.8,.5)and(.5,.3) .. (.3,-.2)..controls (.2,-.7)and(.2,-.7) .. (.12,-5.2);
\draw (-1,.66).. controls (-.8,.5)and(-.5,.3) .. (-.3,-.2)..controls (-.2,-.7)and(-.2,-.7) .. (-.12,-5.2);
\draw (-2.8,0).. controls (-1.5,-.2)and(-1.3,-.3) .. (-1,-.4)..controls (-.8,-.5)and(-.7,-.6) .. (-.7,-.7)..controls (-.8,-.7)and(-.52,-.7) .. (-.52,-2)..controls (-.52,-2.5)and(-.52,-3) .. (-.52,-5);
\end{tikzpicture}}
\hspace{12em}
\subfloat[{Spatial slice of an non-extreme black hole with two asymptotic ends.}]{\label{fig:2(b)}
\begin{tikzpicture}[scale=.4, every node/.style={scale=0.7}]
\draw [fill=gray!80!white,fill opacity=0.5,gray!80!white]plot[smooth,samples=100,variable=\y,domain=-2.5:2.5]({3*\y*\y/4+.5},\y)--(2.583,2.5)--(-2.583,2.5)--plot[smooth,samples=100,variable=\y,domain=2.5:-2.5]({-3*\y*\y/4-.5},\y)--(-2.583,-2.5)--(2.583,-2.5);
\draw[fill=green!30!white] (0,2.5) ellipse (149pt and 32pt);
\draw(0,2.4) ellipse (125pt and 25pt);
\draw(0,2.3) ellipse (105pt and 20pt);
\draw(0,2.2) ellipse (95pt and 16pt);
\draw (0,2) ellipse (80pt and 14pt);
\draw[black] (0,1.9) ellipse (58pt and 10pt);
\draw[black] (0,.75) ellipse (25pt and 3pt);
\draw[red,very thick] (0,0) ellipse (15pt and 3pt);
\draw [fill=gray!80!white,fill opacity=0.5](0,-2.5) ellipse (149pt and 30pt);
\draw (0,-2.3) ellipse (125pt and 25pt);
\draw[black] (0,-2) ellipse (95pt and 14pt);
\draw[black] (0,-1.5) ellipse (60pt and 8pt);
\draw[black] (0,-.75) ellipse (25pt and 3pt);
\draw[->,black] (3,0)--(1,0)node[black,thick,right=2cm,font=\large]{$H$};
\draw plot[smooth,samples=100,variable=\y,domain=-3.15:3.15]({2*\y*\y/5+.3},\y);
\draw plot[smooth,samples=100,variable=\y,domain=-3.15:3.15]({-2*\y*\y/5-.3},\y);
\draw plot[smooth,samples=100,variable=\y,domain=-3.5:3.6]({2*\y*\y/23+.1},\y);
\draw plot[smooth,samples=100,variable=\y,domain=-3.5:3.6]({-2*\y*\y/17-.1},\y);
\end{tikzpicture}}
\caption{Slice of Myers-Perry solution in extreme and usual cases}\label{fig2}
\end{figure} 
%***************************************figure******************************************
\subsection{Biaxisymmetric Initial Data} \label{sec2}
%\section{Extrinsic curvature of initial data and angular momentum}\label{sec2} 
Our goal is to construct initial data which represent deformations of the Myers-Perry initial data discussed above. The strategy, following \cite{dain2012geometric}, is to use the conformal method to  reduce the problem to an elliptic PDE for a single scalar field. First however, we will need to parameterize our initial data sets appropriately, to isolate the functional degrees of freedom.  For the biaxisymmetric data sets under consideration, there is a convenient way to achieve this, which generalizes the approach for axisymmetric three-dimensional initial data sets. 

Consider an asymptotically flat spacetime with an isometry group which admits an $U(1)^2$ subgroup.  We now briefly review the computation of angular momenta from the twist potentials (see \cite{hollands2012horizon} for a general discussion with $U(1)^{D-3}$ rotational symmetries). We will denote the generators of rotational symmetries as $m_i$, $i=1,2$.  The orbits of the $m_i$ have period $2\pi$.  For simplicity, assume there is one asymptotic end and let $S^3_\infty$ represent the sphere at spatial infinity.  We will take $n$ and $s$ to be unit timelike and spacelike vector fields which span the tangent space normal to $S^3_\infty$. As is well known, one can define conserved angular momenta from the Komar integrals
\begin{equation}
J_i = \frac{1}{16\pi} \int_{S^3_{\infty}} \star \td m_i = \frac{1}{16\pi} \int_{S^3_\infty} i_n i_s \td m_i \, \td S
\end{equation} where $\td S$ represents the volume element on $S^3_\infty$.   The vacuum equations imply the following one-forms are closed:
\begin{equation}
\lambda_i = \star (m_1 \wedge m_2 \wedge \td m_i)
\end{equation} and we may therefore define local twist potentials $\omega_i$ satisfying $\td \omega_i = \lambda_i$. We assume that the spacetime is simply connected so the $\omega_i$ are globally defined.   We can then evaluate the $J_i$ in terms of these twist potentials as follows.  Using the fact the $m_i$ generate commuting isometries and applying an interior derivative with respect to $m_i$, we find
%First note that $\star^2 = -1$ on any p-form in $D=5$ Lorenzian spacetime. Hence 
%\begin{equation}
%-\star \lambda_i = m_1 \wedge m_2 \wedge \td m_i
%\end{equation} Now apply $i_{m_2}$ to both sides to find
%\begin{equation}
%g(m_1,m_2)m_2 \wedge \td m_i  - g(m_2,m_2) m_1 \wedge \td m_i - m_1 \wedge m_2 \td(g_{2i}) = -i_{m_2} \star \td \lambda_i 
%\end{equation} where $g(m_i,m_j) = m_i \cdot m_j$ and we have used the fact that $i_{m_2} \td \star \lambda_i + \td i_{m_2} \star \lambda = L_{m_2} \lambda_i = 0$ since $m_2$ is an isometry and commutes with $m_1$.  Next apply $i_{m_1}$ to both sides.  After some algebra we arrive at
%
\begin{eqnarray} \label{KVFcov}
\td m_i &=& \frac{1}{\text{det} g_{ij}}\left[i_{m_1} i_{m_2} \star \lambda_i  + \left(g(m_1,m_2) \td g_{2i} - g(m_2,m_2) \td g_{1i}\right)m_1\right. \nonumber  \\ &+& \left.  \left(g(m_1,m_2) \td g_{1i} - g(m_1,m_1) \td g_{2i}\right) m_2\right]
\end{eqnarray} where $g(m_i,m_j) = m_i \cdot m_j$ and $\text{det}g_{ij} = g(m_1,m_1)g(m_2,m_2) - (g(m_1,m_2))^2$. Thus we find
\begin{equation}
i_{n} i_{s} \td m_i =  \frac{i_n i_s i_{m_1} i_{m_2} \star \td \omega_i}{\text{det} g_{ij}}
\end{equation} Define the one-form 
\begin{equation}\label{xi}
\xi \equiv (\det g_{ij})^{-1/2} \star (n \wedge s \wedge m_1 \wedge m_2) \to \xi_a = (\det g_{ij})^{-1/2} \epsilon_{abcde} n^b s^c m_1^d m_2^e
\end{equation} It can be checked that $\xi$ has unit length and is orthogonal to remaining members of the co-frame $(n,s,m_i)$ (and in particular is  tangent to $S^3_\infty$). We may then define a coordinate $x$ such that $\xi$ is proportional to $\td x$, i.e. $\xi = \sqrt{g_{xx}} \td x$. As discussed in precise detail in \cite{hollands2012horizon}, $x$ parameterizes $S^3_\infty / U(1)^2$ and we normalize it so $-1 \leq x \leq 1$ where $x = \pm 1$ correspond to the poles where the $m_i$ vanish.  We then have $i_n i_s \star \td m_i = (\det g_{ij})^{-1/2} \xi \cdot \td \omega_i$ and so  
\begin{equation}
J_i = \ \frac{1}{16 \pi} \int_{S^3_\infty} (\det g_{ij})^{-1/2} \left(\xi \cdot (\td \omega_i) \right) \, \td S  = \frac{\pi}{4} \int_{-1}^{1} \frac{\partial \omega_i}{\partial x}\, \td x = \frac{\pi}{4} \left(\omega_i(1) - \omega_i(-1)\right)
\end{equation}
%In this section we want to bring a general argument similar to Dain \cite{dain2012geometric} for 5 dimensional biaxisymmetric black, $U(1)^2$, to calculate angular momentum and extrinsic curvature in particular way. 

It is useful in the following to work with respect to the preferred tetrad $(n,s,\xi,m_i)$.  For concreteness we introduce the vector fields $\eta^a = m_1^a$ and $\gamma^a = m_2^b$ with associated scalar products $\eta^a \eta_a = \eta, \eta^a \gamma_a = L, \gamma^a \gamma_a = \gamma$, and $H = \det g_{ij} = \eta \gamma - L^2$.  One can  write the metric in this basis as 
\begin{equation}
g_{ab}=-n_an_b+s_as_b+\xi_a \xi_b + \frac{\gamma}{H}\eta_a\eta_b+\frac{\eta}{H}\gamma_a\gamma_b-\frac{2L}{H}\eta_{(a}\gamma_{b)}
\end{equation}
%Where $n^a$ is unit vector orthogonal to spacelike hypersurface, $r^a$ is unit vector orthogonal to timelike hypersurface \cite{poisson2004relativist}. We have two formulas for angular momentum of spacetime. The first one is Komar integral
%\begin{equation}
%J_{\psi}=\frac{1}{16\pi}\int_{S^3_{\infty}} *d\eta=\frac{1}{16\pi}\int_{S^3_{\infty}} \epsilon_{abcde}\bar{\nabla}^d\eta^e=\frac{1}{16\pi}\int_{S^3_{\infty}}\bar{\nabla}^d\eta^e d\Omega_{de}\label{6}
%\end{equation}
%\begin{equation}
%J_{\varphi}=\frac{1}{16\pi}\int_{S^3_{\infty}} *d\gamma=\frac{1}{16\pi}\int_{S^3_{\infty}} \epsilon_{abcde}\bar{\nabla}^d\gamma^e=\frac{1}{16\pi}\int_{S^3_{\infty}}\bar{\nabla}^d\gamma^e d\Omega_{de}
%\end{equation}
%Where $d\Omega_{de}=-2n_{[d}r_{e]}\sqrt{\sigma}d^3y=-2n_{[d}r_{e]}dS$ and $\bar{\nabla}$ is connection respect to Lorentzian metric $g$. The twist one form also obtain from 
%\begin{eqnarray}
%(\omega_{\psi})_a&=&*(\eta\wedge\gamma\wedge d\eta)=\epsilon_{abcde}\eta^b\gamma^c\bar{\nabla}^d\eta^e\\
%(\omega_{\varphi})_a&=&*(\eta\wedge\gamma\wedge d\gamma)=\epsilon_{abcde}\eta^b\gamma^c\bar{\nabla}^d\gamma^e
%\end{eqnarray}
%We write integrand of Komar integral in form of twist potential. Then if $(\omega_{\psi})_a=\omega_a=d\omega$ we have
In this notation, the expression \eqref{KVFcov} is
\begin{eqnarray}
\bar{\nabla}^k\eta^l&=&\frac{1}{2H}\epsilon^{aijlk}\eta_i\gamma_j{\lambda_1}_a+\frac{2}{H}\gamma_j\bar{\nabla}^{[k}\eta^{|j|}P^{l]}+\frac{1}{H}\bar{\nabla}^{[k}\eta B^{l]}+\frac{1}{H}\gamma_j\bar{\nabla}^j\eta\eta^{[l}\gamma^{k]}\label{eta}\\
\bar{\nabla}^k\gamma^l&=&\frac{1}{2H}\epsilon^{aijlk}\eta_i\gamma_j{\lambda_2}_a+\frac{1}{H}\bar{\nabla}^{[k}\gamma P^{l]}+\frac{2}{H}\eta_j\bar{\nabla}^{[k}\gamma^{|j|} B^{l]}+\frac{1}{H}\gamma_j\bar{\nabla}^j\eta\eta^{[l}\gamma^{k]}\label{gamma}
\end{eqnarray}
where $P^k\equiv \eta\gamma^k-L\eta^k$, and $B^k\equiv\eta^k\gamma-L\gamma^k$. %Moreover, we have 
%\begin{equation}
%\eta^an_a=\gamma^an_a=\eta^ar_a=\gamma^ar_a=0\Rightarrow P^an_a=S^an_a=P^ar_a=S^ar_a=0
%\end{equation}
%Therefore, we have 
%\begin{equation}
%\bar{\nabla}^k\eta^ln_kr_l=\frac{1}{2\sqrt{H}}\xi^a\omega_a\quad\quad\quad\text{where}\quad\quad\quad \xi^a:=\frac{1}{\sqrt{H}}\epsilon^{aijkl}\eta_i\gamma_jn_kr_l\label{5}
%\end{equation}
%The coordinate $\theta$ is such that the unit vector $\xi^a$ is proportional to $\theta^a=\left(\frac{\partial}{\partial\theta}\right)^a$. So the angular momentum will be
%\begin{eqnarray}
%J_{\psi}&=&\frac{1}{16\pi}\int_{S^3_{\infty}} *d\eta\nonumber\\
%&=&\frac{1}{16\pi}\int_{S^3_{\infty}}\frac{1}{\sqrt{H}}\xi^a\omega_adS\nonumber\\
%&=&-\frac{1}{8\pi}\int_{S^3_{\infty}}S^{2}_ar^a dS\nonumber\\
%&=&\frac{\pi}{4}\int_0^{\pi/2}\partial_{\theta}\omega d\theta\nonumber\\
%&=&\frac{\pi}{4}\left(\omega\left(\frac{\pi}{2}\right)-\omega\left(0\right)\right)\nonumber
%\end{eqnarray}
%In first equality we used (\ref{5}) and (\ref{6}) and in secend equality we used definition of $S^2_{a}$ in (\ref{4}). In third equality we used the fact that $\xi$ is a vector field in direction of $\theta$. We can do similar calculation for $J_{\varphi}$. Thus angular momentum of five dimensional black hole can be expressed as the value of twist potential at axis, i.e. the Killing part of spacetime metric in Weyl coordinate vanishes.\\ \par \noindent
As in four dimensions, we will now establish sufficient conditions under which one can construct the extrinsic curvature from potentials.  We will restrict attention to initial data that are invariant under the biaxial $U(1)^2$ symmetry, i.e. $(\Sigma,h)$ admits an $U(1)^2$ acting as isometries and $L_{m_i} K=0$ where once again we denote the generators of the rotational symmetries by  $m_i$.  We are of course interested in spacetimes with $U(1)^2$ isometry, and one can always find initial data surfaces with this symmetry; conversely, given initial data with these symmetries, the evolution will preserve the symmetry, although they may be `hidden' \cite{KIDS}. To begin, it is easiest to introduce a tetrad for the metric on the spacelike slice such that  (in the rest of this note, we will use Latin indices $a,b \ldots$ to run over $1\dots 4$) 
\begin{equation}
h_{ab} = q_{ab} + h_{ij}m^i_a m^j_b
\end{equation} where $q_{ab} = s_a s_b + \xi_a \xi_b$ is the part of the metric orthogonal to the surfaces of transitivity of the $U(1)^2$ action and  $h_{ij} = m_i \cdot m_j$. In terms of the spacetime tetrad presented earlier, one can think of $h_{ab}$ as the metric induced on the surface with normal $n$, i.e. $h_{ab} = g_{ab} + n_a n_b$.  We then have
\begin{eqnarray}
h =q_{ab}+\frac{\gamma}{H}\eta_a\eta_b+\frac{\eta}{H}\gamma_a\gamma_b-\frac{2L}{H}\eta_{(a}\gamma_{b)}
\end{eqnarray}   Now consider a maximal slice $\textrm{tr} K =0$ with a given extrinsic curvature tensor $K_{ab}$. The square of this tensor, which appears in the constraint equations \eqref{eq:constraints}, can be computed in the frame defined above:
\begin{eqnarray}
K_{ab}K^{ab}&=&K_{ab}K_{cd}h^{ac}h^{bd}\nonumber\\
%&=&{K_{ab}K_{cd}q^{ac}q^{bd}}+{K_{ab}K_{cd}m^{ac}m^{bd}}+{K_{ab}K_{cd}m^{ac}q^{bd}}
%+{K_{ab}K_{cd}q^{ac}m^{bd}}\nonumber\\
&=&{K_{ab}K_{cd}q^{ac}q^{bd}}+\left(K_{ab}\eta^a\eta^b\right)^2\frac{\gamma^2}{H^2}+\left(K_{ab}\gamma^a\gamma^b\right)^2\frac{\eta^2}{H^2}+\left(K_{ab}\eta^a\eta^b\right)\left(K_{ab}\gamma^a\gamma^b\right)\frac{4L^2}{H^2}\nonumber\\
&+&\frac{2\gamma\eta}{H^2}\left(K_{ab}\eta^a\gamma^b\right)^2
-\frac{2L\gamma}{H^2}\left(K_{ab}\eta^a\eta^b\right)\left(K_{cd}\eta^c\gamma^d\right)
-\frac{2L\eta}{H^2}\left(K_{ab}\eta^a\gamma^b\right)\left(K_{cd}\gamma^c\gamma^d\right)\nonumber\\
&+&\frac{2\gamma}{H}S^2_aS^{2a}+\frac{2\eta}{H}S^1_aS^{1a}-\frac{4L}{H}S^1_aS^{2a}\label{3}
\end{eqnarray}
where the $S^i_a$ are defined as 
\begin{eqnarray}
S^1_a&\equiv& K_{ab}\eta^b+\frac{L}{H}\gamma_a K_{cd}\eta^c\eta^d-\frac{\gamma}{H}\eta_a K_{cd}\eta^c\eta^d+\frac{L}{H}\eta_a K_{cd}\gamma^c\eta^d-\frac{\eta}{H}\gamma_a K_{cd}\gamma^c\eta^d\label{S2}\\
S^2_a&\equiv& K_{ab}\gamma^b+\frac{L}{H}\eta_a K_{cd}\gamma^c\gamma^d-\frac{\eta}{H}\gamma_a K_{cd}\gamma^c\gamma^d+\frac{L}{H}\gamma_a K_{cd}\gamma^c\eta^d-\frac{\gamma}{H}\eta_a K_{cd}\gamma^c\eta^d,\label{S1}
\end{eqnarray}
One can easily check that $S^1_a$ and $S^2_a$ are each orthogonal to  $\gamma^a$ and $\eta^a$  and they are invariant under Lie derivatives with respect to $\gamma^a$ and $\eta^a$.  Furthermore, because $K_{ab}$ satisfies the momentum constraint, $\nabla_a K^{ab} =0$, one can deduce that the $S^i$ are divergenceless; $\td \star S^i =0$. 
Now define the one-forms
\begin{equation}
\mathcal{K}^i = \star (S^i \wedge m_2 \wedge m_1)
\end{equation} which in our basis take the form
\begin{gather}
\mathcal{K}^1_{a}=\epsilon_{abcd}S^{1b}\gamma^c\eta^d,\quad\quad \mathcal{K}^2_{a}=\epsilon_{abcd}S^{2b}\gamma^c\eta^d
\end{gather} We note these forms are closed, i.e. 
\begin{equation}
\td \mathcal{K}^i = i_{m_2} i _{m_1} \td \star S^i = 0
\end{equation}  We then define the potentials $\bar\omega_i$ by 
\begin{equation}
\mathcal{K}^i = -\frac{\td \bar{\omega}_i}{2}
\end{equation} These $\bar\omega_i$ are in fact the pullback to $\Sigma$ of the spacetime twist potentials $\omega_i$ defined in the previous section. The proof of this statement is similar to the three dimensional case given in \cite{dain2008axisymmetric}, using the expression \eqref{eta}.  Hence as we are working at the level of initial data we will simply use $\omega_i$ from now on.   One can invert these expressions to find
\begin{equation}
S^i =  \frac{1}{H} i_{m_2} i_{m_1} \star \mathcal{K}^i = \frac{1}{2H} i_{m_1} i_{m_2} \star \td{\omega}^i
\end{equation} In our tetrad basis, 
\begin{equation}
S^1_a = -\frac{1}{H}\epsilon_{a b c d} \gamma^b \eta^c \mathcal{K}^{1d} = \frac{1}{2H}\epsilon_{a b c d} \gamma^b \eta^c \td {\omega}^{1d}   \qquad S^2_a = -\frac{1}{H}\epsilon_{a b c d} \gamma^b \eta^c \mathcal{K}^{2d} = \frac{1}{2H}\epsilon_{a b c d} \gamma^b \eta^c \td {\omega}^{2d}
\end{equation} 
Now we define a symmetric, divergence free, and trace free tensor field
\begin{equation}
\bar{K}_{ab}:=\frac{2}{H}\Bigg[\left(\eta S^2{}_{(a}\gamma{}_{b)}-LS^1{}_{(a}\gamma{}_{b)}\right)+\left(\gamma S^1{}_{(a}\eta{}_{b)}-LS^2{}_{(a}\eta{}_{b)}\right)\Bigg].\label{1}
\end{equation}
Then the full contraction of this tensor is
\begin{equation}
\bar{K}_{ab}\bar{K}^{ab}=\frac{2}{H}\Bigg[\gamma S^1_aS^{1a}+\eta S^2_aS^{2a}-2LS^1_aS^{2a}\Bigg]
\end{equation}
Now assume the initial data has $t-\phi^i$ symmetry as defined in \cite{figueras2011black}. Here  $\phi^i$ are coordinates adapted to the commuting Killing fields $m_i$.  Then:  1) $\partial / \partial \phi^i$ are Killing vector generator of $U(1)^2$ isometry group of $(\Sigma,h_{ab})$, and 2) $\phi^i\rightarrow -\phi^i$ is a diffeomorphism which preserves $h_{ab}$ but reverses the sign of $K_{ab}$. This is equivalent to the following conditions
\begin{enumerate}
\item $K_{ab}\eta^a\eta^b=K_{ab}\gamma^a\gamma^b=K_{ab}\gamma^a\eta^b=0$
\item $K_{ab}q^{ac}q^{bd}=0$
\end{enumerate}
By condition (1)  all terms of (\ref{3}) are zero except the last three ones, also by definition of $S^i_a$ we have 
\begin{equation}
S^1_a=K_{ab}\gamma^b\quad\quad\quad S^2_a=K_{ab}\eta^b
\end{equation}
Moreover, by condition (2) we have
\begin{eqnarray}
0&=&K_{ab}q^{ac}q^{bd}\nonumber\\
&=&K_{ab}\left(h^{ac}-\frac{\gamma}{H}\eta^a\eta^c-\frac{\eta}{H}\gamma^a\gamma^c+\frac{2L}{H}\eta^a\gamma^c\right)\left(h^{bd}-\frac{\gamma}{H}\eta^b\eta^d-\frac{\eta}{H}\gamma^b\gamma^d+\frac{2L}{H}\eta^b\gamma^d\right)\nonumber\\
%&=&K^{cd}-\frac{\gamma}{H}K^c{}_b\eta^b\eta^d-\frac{\eta}{H}K^c{}_b\gamma^b\gamma^d+\frac{2L}{H}K^c{}_b\eta^b\gamma^d-\frac{\gamma}{H}K_a{}^d\eta^a\eta^c-\frac{\eta}{H}K_a{}^d\gamma^a\gamma^c+\frac{2L}{H}K_a{}^d\eta^a\gamma^c\nonumber\\
%&=&K^{cd}-\frac{\gamma}{H}\left(K^c{}_b\eta^b\eta^d+K_a{}^d\eta^a\eta^c\right)-\frac{\eta}{H}\left(K^c{}_b\gamma^b\gamma^d+K_a{}^d\gamma^a\gamma^c\right)+\frac{2L}{H}K^c{}_b\eta^b\gamma^d+\frac{2L}{H}K_a{}^d\eta^a\gamma^c\nonumber\\
&=&K^{cd}-\frac{2}{H}\Bigg[\left(\eta S{}^{2(c}\gamma{}^{d)}-LS{}^{1(c}\gamma{}^{d)}\right)+\left(\gamma S{}^{1(c}\eta{}^{d)}-LS{}^{2(c}\eta{}^{d)}\right)\Bigg]\nonumber\\
&=&K^{cd}-\bar{K}^{cd}
\end{eqnarray}
Then we see that an arbitrary maximal $t-\phi^i$-symmetric extrinsic curvature tensor can be constructed from the twist potentials $\omega_i$:
\begin{equation}
K_{ab}=\bar{K}_{ab}\label{2}
\end{equation} We emphasize that in general, one cannot construct the complete extrinsic curvature tensor directly from twist potentials. 
 Although we are not going to use the following, it is interesting to see the relationship between the form of the extrinsic curvature given in (\ref{2}) and the expression given in \cite{figueras2011black}, valid in $t-\phi^i$ symmetry: 
\begin{equation}
K_{ab}=J^i{}_{(a}\phi^i{}_{b)}\quad\quad i=1,2 
\end{equation}
(see also  \cite{0264-9381-31-5-055004}).  We then have
\begin{equation}
J^1_a=\frac{2}{H}\left(\eta S^2_a-LS^1_a\right),\quad\quad J^2_a=\frac{2}{H}\left(\gamma S^1_a-LS^2_a\right)
\end{equation}

\section{Main result}
The classical method to prove the existence of solutions to the constraint equations is the conformal method.  In the case of extreme Myers-Perry with initial data we defined in section \ref{sec1} $(\Sigma,h_{ab},K_{ab})$ we can write the metric in the following conformal form 
\begin{equation}
h_{ab}=\Phi_0^2\, \tilde{h}_{ab}.
\end{equation}
with 
\begin{eqnarray}
\tilde{h}=e^{2U}\left(\td\rho^2+\td z^2\right)+\sigma'_{ij}\td\phi^i\td\phi^j\label{htild}
\end{eqnarray}
where $\rho=\frac{1}{2}r^2\sin 2\theta$ ,$z=\frac{1}{2}r^2\cos 2\theta$, $\det\sigma'_{ij}=\rho^2$, $\phi^1=\varphi$, and $\phi^2=\psi$. The conformal factor $\Phi_0$ and the functions $U$ and $\sigma'_{ij}$ in the metric are defined and studied in Appendix \ref{AppC}.  We may also write
\begin{equation}
K_{ab}=\Phi^{-2}_0\tilde{K}_{ab}.
\end{equation} and by section \ref{sec2},  $t-\phi^i$ symmetry of the  initial data implies we may express the second factor as
\begin{equation}
\tilde{K}_{ab}=\frac{2}{\rho^2}\Bigg[\left(\sigma'_{22} \tilde{S}^1{}_{(a}\eta{}_{b)}-\sigma'_{12}\tilde{S}^2{}_{(a}\eta{}_{b)}\right)+\left(\sigma'_{11} \tilde{S}^2{}_{(a}\gamma{}_{b)}-\sigma'_{12}\tilde{S}^1{}_{(a}\gamma{}_{b)}\right)\Bigg]\label{ktild}
\end{equation}
where $\gamma^a=\left(\frac{\partial}{\partial\varphi}\right)^a$, $\eta^a=\left(\frac{\partial}{\partial\psi}\right)^a$ and
\begin{gather}
\tilde{S}^1_{a}=\frac{1}{2\rho^2}\tilde{\epsilon}_{abcd}\gamma^b\eta^c\tilde{\nabla}^d\omega_{\varphi},\quad\quad \tilde{S}^2_{a}=\frac{1}{2\rho^2}\tilde{\epsilon}_{abcd}\eta^b\gamma^c\tilde{\nabla}^d\omega_{\psi}
\end{gather}
where the twist potentials $\omega_i$ are given in Appendix \ref{AppC} and $\tilde{\epsilon}_{abcd}$ and $\tilde{\nabla}$ are respectively the volume element and the connection  associated to $\tilde{h}_{ab}$. This particular form of extrinsic curvature implies $\tilde{\nabla}_{a}\tilde{K}^{ab}=0$ and so it satisfies \eqref{eq:constraints2}.  We will perturb about this solution (i.e. we freely specify variations of the functions appearing in the metric $\tilde{h}$ with appropriate fall-off behaviour) and demonstrate the existence of a conformal factor $\Phi$ that solves  the constraint equations, yielding a new family of initial data $(\Sigma, h_{ab}, K_{ab})$ where $\Phi_0$ is replaced by $\Phi$ above.  More precisely, our main result is
\begin{thm}\label{maintheorem} Let $(\Sigma,h_{ab},K_{ab})$ be the maximal biaxisymmetric initial data set of extreme Myers-Perry described in section \ref{sec1} with angular momenta $J_{\varphi}$ and $J_{\psi}$ and mass $M$. Then there is a small $\lambda_0$ such that for $-\lambda_0<\lambda<\lambda_0$ there exists a family of initial datasets $(\Sigma,h_{ab}^{\lambda},K_{ab}^{\lambda})$ (i.e. solutions of the constraints on $\Sigma$) such that:
\begin{enumerate}
\item For $\lambda=0$ the family of initial data is that of extreme Myers-Perry initial data, i.e. $(\Sigma,h_{ab},K_{ab})$. The family is differentiable in $\lambda$ and it is close to extreme Myers-Perry with respect to an appropriate norm which involves two derivatives of the metric.
\item The data have the same asymptotic geometry as the extreme Myers-Perry initial dataset. The angular momenta and the area of the cylindrical end in the family do not depend on $\lambda$; they have same value as in $(\Sigma,h_{ab},K_{ab})$, namely $J_{\varphi}$, $J_{\psi}$ and $A_0$ , respectively.
\item The family of data are biaxisymmetric and maximal (i.e $\text{tr}K^{\lambda}=0$).
\end{enumerate}
\end{thm}
\noindent An important parameter of an initial data set with a cylindrical end is the area of a cross-section. If $A(r)$ is the area of constant $r$, we have
\begin{equation}
A_0=\lim_{r\rightarrow 0} A(r)=2\pi^2\mu^2\sqrt{ab}.
\end{equation}
This corresponds to the area of the event horizon of the corresponding extreme Myers-Perry black hole.  Consider a member of the family of initial data set $(\Sigma,h_{ab}^{\lambda},K_{ab}^{\lambda})$ for fixed $\lambda\neq 0$.  By an argument similar to that given in \cite{Dain:2010uh}, the fall-off of the lapse and shift can always be selected so that the geometry of the cylindrical end and its area will be preserved, for sufficiently short times.

\section{Proof of main result} We now turn to the derivation of the result discussed in the previous section. 
\begin{proof}
Let $(\Sigma,h_{ab},K_{ab})$ be the maximal initial data set (given in Appendix C) of the extreme Myers-Perry black hole. These satisfy the constraint equations:
\begin{gather}
R-K_{ab}K{}^{ab}=0,  \label{vaccons1}\\
\nabla^aK_{ab}=0. \label{vaccons2}
\end{gather}
To construct a solution of these constraint equations we use classical conformal method with rescaling
\begin{equation}
h_{ab}=\Phi^2_0\tilde{h}_{ab},\quad\quad K_{ab}=\Phi^{-2}_0\tilde{K}_{ab}.
\end{equation}
where $\tilde{h}_{ab}$ and $\tilde{K}_{ab}$ are defined in equations (\ref{htild}) and (\ref{ktild}), respectively. In conformal data the constraint equations are 
\begin{gather}
\Delta_{\tilde{h}}\Phi_0-\frac{1}{6}\tilde{R}\Phi_0 +\frac{1}{6}\tilde{K}_{ab}\tilde{K}^{ab}\Phi_0^{-5}=0.\label{32}\\
\tilde{\nabla}_b\tilde{K}^{ab}=0. \label{15}
\end{gather}
By construction $\tilde{K}_{ab}$ in section \ref{sec2} is always divergence-free and traceless, so the momentum constraint equation (\ref{15}) is automatically satisfied and we need only consider the Lichnerowicz equation (\ref{32}). The Laplace operator associated with the metric \eqref{htild} (for any $U,\sigma'_{ij}$) in biaxial symmetry can be written 
\begin{eqnarray}
\Delta_{\tilde{h}}\Phi=\frac{e^{-2U}}{r^2}\Delta_4\Phi
\end{eqnarray}  
Where $\Phi$ is an arbitrary function of only $r$ and $\theta$ and $\Delta_4$ is the flat four dimensional Laplace operator respect to metric 
\begin{eqnarray}
\delta_4&=&\frac{1}{2\sqrt{\rho^2+z^2}}\left(\td\rho^2+\td z^2\right)+\left(\sqrt{\rho^2+z^2}-z\right)\td\varphi+\left(\sqrt{\rho^2+z^2}+z\right)\td\psi^2\nonumber\\
&=&\td r^2+r^2\td\theta^2+r^2\sin^2\theta\td\varphi+r^2\cos^2\theta\td\psi^2
\end{eqnarray}
The scalar curvature of the metric \eqref{htild} is
\begin{eqnarray}\label{curv1}
\tilde{R}=e^{-2U}\left( -2 \Delta_2 U +\frac{\det\td\sigma'}{2\rho^2}\right)\equiv\frac{e^{-2U}}{r^2}\tilde{R}_0
\end{eqnarray}
where $\Delta_2$ is the Laplacian with respect to the flat 2 dimensional metric i.e. $\delta_2=\td\rho^2+\td z^2$. The extrinsic curvature is 
\begin{eqnarray}\label{K1}
\tilde{K}_{ab}\tilde{K}^{ab}&=&\frac{e^{-2U}}{2\rho^4}\Bigg[\sigma'_{11}(\td\omega_{\psi})^2+\sigma'_{22}(\td\omega_{\varphi})^2-2\sigma'_{12}(\td\omega_{\varphi})\cdot(\td\omega_{\psi})\Bigg]\equiv\frac{e^{-2U}}{r^2}\tilde{K}_0^2
\end{eqnarray}
where $\cdot$ is the inner product with respect to $\delta_2$. Then the Lichnerowicz equation  \eqref{32} for the conformal triple $(\Sigma,\tilde{h}_{ab},\tilde{K}_{ab})$ is
\begin{eqnarray}
\Delta_4\Phi_0-\frac{\tilde{R}_0}{6}\Phi_0+\frac{\tilde{K}^2_0}{6\Phi_0^5}=0\label{33}
\end{eqnarray}

\noindent We now perturb equation (\ref{33}) about the solution given by the maximal initial data for the extreme Myers-Perry black hole by taking
\begin{eqnarray}
U&\rightarrow& U+\lambda U_1\nonumber\\
\sigma'_{ij}&\rightarrow& \sigma'_{ij}+\lambda \bar{\sigma}_{ij}\nonumber\\
\omega_{\varphi}&\rightarrow& \omega_{\varphi}+\lambda \omega_{1}\nonumber\\
\omega_{\psi}&\rightarrow& \omega_{\psi}+\lambda \omega_{2}\label{34}
\end{eqnarray}
for a fixed set of $U(1)^2$-invariant functions $U_1$, $\bar{\sigma}_{ij}$, $\omega_{1}$, $\omega_{2}$, and small $\lambda$, and then seek a solution $\Phi$ of the form 
\begin{equation}
\Phi=\Phi_0+u.\label{35}
\end{equation} where $u$ is a function to be determined. 
Inserting (\ref{34}) and (\ref{35}) into (\ref{33}), we have
\begin{equation}
\mathcal{E}(\lambda,u)=0.\label{36}
\end{equation}
where
\begin{eqnarray}
\mathcal{E}(\lambda,u)&=&\Delta_4\left(\Phi_0+u\right)-\frac{1}{6}\tilde{R}_{\lambda}\left(\Phi_0+u\right)+\frac{\tilde{K}^{2}_{\lambda}}{6(\Phi_0+u)^5}\label{47}
\end{eqnarray}
where $\tilde{R}_{\lambda}$ and  $\tilde{K}^{2}_{\lambda}$ are obtained from $\tilde{R}_0$ and  $\tilde{K}_0$  using the transformation (\ref{34}). If we plug in $\lambda=0$ in equation ($\ref{47}$), we have equation ($\ref{33}$).  Then to prove theorem \ref{maintheorem}, it is enough to show existence and uniqueness of the solution of equation (\ref{36}) and this will be done in the next lemma. We then obtain a family of solutions $(\Sigma,h^\lambda_{ab}, K^\lambda_{ab})$ to \eqref{vaccons1} and \eqref{vaccons2} with $h^\lambda_{ab} = \Phi^2 \tilde{h}^\lambda_{ab}$ and $K^\lambda_{ab} = \Phi^{-2} \tilde{K}^\lambda_{ab}$. 
\end{proof}
\begin{lemma} Let $\omega_1,\omega_2\in C^{\infty}_c(\mathbb{R}^4\backslash \Gamma)$ and ${U}_1,\bar{\sigma}_{ij}\in C^{\infty}_c(\mathbb{R}^4\backslash \{0\})$. Then, there exists  $\lambda_0>0$ such that for all $\lambda\in(-\lambda_0,\lambda_0)$
\begin{enumerate}
\item There exists a solution $u(\lambda)$  of (\ref{36}) belonging to $ W^{'2,3}_{-1}$. (for clarify we suppress the $r-$ and $\theta-$ dependence of $u(\lambda)$). 
\item $u(\lambda)$ is continuously differentiable in $\lambda $ and $\Phi_0+u(\lambda)>0$.
\item $u(\lambda)$ is the unique solution of (\ref{36}) for small $u$ and small $\lambda$.
\end{enumerate}
\end{lemma}
\begin{remark} 
Here $\Gamma\equiv\{\rho=0\}$ is the axis on which the  Killing part of \eqref{htild}  becomes degenerate (i.e. at least one combination of $\partial/\partial \psi$ and $\partial / \partial \phi$ vanishes).
\end{remark}
\noindent Here $ W^{'2,3}_{-1}$ is one of Bartnik's weighted Sobolev spaces (appendix \ref{appA}). This space is consistent with the desired fall-off conditions of the solution $u$ at the cylindrical end and asymptotically flat end. Moreover, we do not expect $u$ to be regular at the origin. By section \ref{sec2} we know the angular momenta are equal to the difference to potentials evalauted on the endpoints of the axis (parameterized here by $\theta$). Therefore, with the requirement that $\omega_1,\omega_2\in C^{\infty}_c(\mathbb{R}^4\backslash \Gamma)$ the angular momenta are preserved by the family of deformations, which implies part 2 of Theorem \ref{maintheorem}.
%*************************************************************
\subsection{Proof of Lemma }
%*************************************************************
The main tool we use to establish the Lemma is the implicit function theorem (see Appendix B of  \cite{Dain:2010uh}).  The argument closely parallels that given in \cite{Dain:2010uh} and proceeds as follows.  Firstly, we select appropriate Banach spaces $X$,$Y$, and $Z$ as required for the implicit function theorem. Then we find neighbourhoods $O_x\subset X$ and $O_y\subset Y$ for which the map $\mathcal{E}:O_x\times O_y\rightarrow Z$ is well-defined. Care must be given to select Banach spaces that satisfy the fall-off conditions on the functions $U$, $\sigma_{ij}$, $\Phi_0$, $\omega_{\varphi}$, and $\omega_{\psi}$ at infinity and singular behavior at the origin of the function  $\Phi_0$. Since the solution need not be regular at the origin  (we are working on $\mathbb{R}^4 - \{0\}$) we cannot select the standard weighted Sobolev spaces $W^{2,3}_{-1}$. To begin we verify that $\mathcal{E}:O_x\times O_y\rightarrow Z$ is $C^1$ .
% The present of singular behavior of functions $f$,$U$, $V$, $W$, $\Phi_0$, $\omega_1$, and $\omega_2$ prevent us to use chain rule for proving that $G$ is $C^1$. 
Next we show that  $D_2\mathcal{E}(0,0)$ (which is defined in equation (\ref{D2E})) is an isomorphism between $Y$ and $Z$.  The implicit function theorem is then used to conclude the existence of a uqniue $u$ with the properties of the lemma. 
%*************************************************************
\subsubsection{$\mathcal{E}$ is well-defined}
%*************************************************************
We choose $X=\mathbb{R}$, $Y=W^{'2,3}_{-1}$ and $Z=L^{'3}_{-3}$. Moreover, we choose $O_x=\mathbb{R}$ and $O_y=\{u\in W^{'2,3}_{-1}: \norm{u}_{W^{'2,3}_{-1}}<\xi\}$ where $\xi$ is computed as follows: by the inequality in Lemma \ref{lem4}-3 for $u\in O_y$ we have
\begin{equation}
r|u|\leq C_0\xi.\label{42}
\end{equation}
where $C_0$ is a constant. Also by lemma \ref{lem3},  we have
\begin{equation}
r\Phi_0\geq (ab\mu)^{1/4}.\label{41}
\end{equation}
Then , if we choose $\xi$ such that 
\begin{equation}
\frac{(ab\mu)^{1/4}}{C_0}>\xi>0\, ,
\end{equation} then for all $u\in O_y$ we will have
\begin{equation}
0<(ab\mu)^{1/4}-C_0\xi\leq r\left(\Phi_0+u\right).\label{38}
\end{equation} 
First we prove that $\mathcal{E}:\mathbb{R}\times O_y\rightarrow L^{'3}_{-3}$ is well-defined. That is, we need to show for $\lambda\in\mathbb{R}$ and $u\in O_y$ we have $\mathcal{E}(\lambda,u)\in L^{'3}_{-3}$. By using the triangle inequality for equation (\ref{36}), we have
\begin{eqnarray}
\norm{ \mathcal{E}(\lambda,u)}_{L^{'3}_{-3}}&\leq& \unb{\norm{\Delta_4 u}_{L^{'3}_{-3}}}_{I}+\unb{\norm{\Delta_4\Phi_0}_{L^{'3}_{-3}}}_{II}+\frac{1}{6} \unb{\norm{\tilde{R}_{\lambda}(\Phi_0+u)}_{L^{'3}_{-3}}}_{III}+\unb{\norm{\frac{\tilde{K}^{2}_{\lambda}}{6(\Phi_0+u)^5}}_{L^{'3}_{-3}}}_{IV}\nonumber\\
\end{eqnarray}
We will show each of these terms are bounded in $L^{'3}_{-3}$. To show this we will need the required properties of the functions ${\omega}_1, {\omega}_2, {U}_1$  and $\bar{\sigma}_{ij}$, as well as the particular fall-off conditions on functions (i.e $U,\sigma'_{ij}$) of the conformal Myers-Perry metric. \\
(I) Since $u\in O_y$
\begin{eqnarray}
\norm{\Delta_4 u}_{L^{'3}_{-3}}\leq\norm{u}_{W^{'2,3}_{-1}}\leq C
\end{eqnarray}
where $C$ is function of $a$ and $b$. Henceforth, the notation $C$ is a constant related only on metric parameters, i.e.  $a$ and $b$.\\
(II) In second term we use the bound on the Laplace operator lemma C.1-\ref{10}:
\begin{equation}
\norm{\Delta_4\Phi_0}_{L^{'3}_{-3}}\leq \norm{\frac{C}{r^6}}_{L^{'3}_{-3}}\leq C
\end{equation}
Finally, since $\omega_1$ and $\omega_2$ have compact support outside the axis and ${U}_1$ and $\bar{\sigma}_{ij}$ have compact support outside the origin, and by using (\ref{38}) and lemma \ref{lem5} one can show that (III) and (IV)) are bounded. The details are tedious and we omit them here.  
Thus $\mathcal{E}:\mathbb{R}\times O_y\rightarrow L^{'3}_{-3}$ is well-defined.
%*************************************************************
\subsubsection{$\mathcal{E}$ is $C^1$ }
%*************************************************************
We denote by $D_1\mathcal{E}(\lambda,u)$ the partial Fr\'echet derivative of $\mathcal{E}$ with respect to the first argument evaluated at $(\lambda,u)$ and by $D_2\mathcal{E}(\lambda,u)$  the partial Fr\'echet derivative of $\mathcal{E}$ with respect to the second argument $u$. These operators are formally obtained by directional derivatives of $\mathcal{E}$ and they are linear operators between the following spaces:
\begin{gather}
D_1\mathcal{E}(\lambda,u):\mathbb{R}\rightarrow L^{'3}_{-3},\\
D_2\mathcal{E}(\lambda,u):W^{'2,3}_{-1}\rightarrow L^{'3}_{-3}.
\end{gather} 
We use the notation $D_1\mathcal{E}(\lambda,u)[\zeta]\in L^{'3}_{-3}$ to denote the operator $D_1\mathcal{E}(\lambda,u)$ acting on $\zeta\in\mathbb{R}$. Similarly, $D_2\mathcal{E}(\lambda,u)[v]\in L^{'3}_{-3}$ denotes the operator $D_2\mathcal{E}(\lambda,u)$ acting on $v\in W^{'2,3}_{-1}$. These linear operators will be
\begin{eqnarray}
D_1\mathcal{E}(\lambda,u)[\zeta]&=&\frac{d}{dt}\mathcal{E}(\lambda+t\zeta,u)\rvert_{t=0}=\frac{1}{6}\Bigg(-D_1\tilde{R}_{\lambda}(\Phi_0+u)+\frac{D_1\tilde{K}^{2}_{\lambda}}{(\Phi_0+u)^5}\Bigg)\zeta,\nonumber
\end{eqnarray}
\begin{eqnarray}
D_2\mathcal{E}(\lambda,u)[v]&=&\frac{d}{dt}\mathcal{E}(\lambda,u+tv)\rvert_{t=0}=\Delta_4 v-\frac{1}{6}\Bigg(\tilde{R}_{\lambda}+\frac{5\tilde{K}^{2}_{\lambda}}{(\Phi_0+u)^6}\Bigg)v\label{D2E}
\end{eqnarray} 
Now, we will prove that the map $\mathcal{E}:\mathbb{R}\times O_y\rightarrow L^{'3}_{-3}$ is $C^1$. As a result of the properties of functions of the metric, we cannot use the chain rule. Alternatively, we will show that:
\begin{enumerate}
\item The linear operator $D_1\mathcal{E}(\lambda,u)[\zeta]$ and $D_2\mathcal{E}(\lambda,u)[v]$ are bounded. i.e.
\begin{gather}
\norm{D_1\mathcal{E}(\lambda,u)[\zeta]}_{L^{'3}_{-3}}\leq C|\zeta|,\label{9}\\
\norm{D_2\mathcal{E}(\lambda,u)[v]}_{L^{'3}_{-3}}\leq C\norm{v}_{W^{'2,3}_{-1}}.
\end{gather}
\item The linear operator $D_1\mathcal{E}(\lambda,u)[\zeta]$ and $D_2\mathcal{E}(\lambda,u)[v]$ are continuous in $(\lambda,u)$ in the operator norms. That is, for every $\epsilon>0$ there exists $\delta>0$ such that
\begin{gather}
|\lambda_1-\lambda_2|<\delta\Longrightarrow \norm{D_1\mathcal{E}(\lambda_1,u)-D_1\mathcal{E}(\lambda_2,u)}_{B(X,Z)}<\epsilon,\\
\norm{u_1-u_2}_{W^{'2,3}_{-1}}<\delta\Longrightarrow \norm{D_2\mathcal{E}(\lambda,u_1)-D_2\mathcal{E}(\lambda,u_2)}_{B(Y,Z)}<\epsilon.
\end{gather}
\item The operators $D_1\mathcal{E}(\lambda,u)[\zeta]$ and $D_2\mathcal{E}(\lambda,u)[v]$ are the partial Fr\'echet derivatives of $\mathcal{E}$. That is
\begin{gather}
\lim_{\zeta\rightarrow 0}\frac{\norm{\mathcal{E}(\lambda+\zeta,u)-\mathcal{E}(\lambda,u)-D_1\mathcal{E}(\lambda,u)[\zeta]}_{L^{'3}_{-3}}}{|\zeta|}=0,\label{30}\\
\lim_{v\rightarrow 0}\frac{\norm{\mathcal{E}(\lambda,u+v)-\mathcal{E}(\lambda,u)-D_2\mathcal{E}(\lambda,u)[v]}_{L^{'3}_{-3}}}{\norm{v}_{W^{'2,3}_{-1}}}=0.\label{31}
\end{gather} 
\end{enumerate}
1. To prove inequality (\ref{9}) we use triangle inequality, lemma \ref{lem5}, and inequality (\ref{38}) then
\begin{eqnarray}
\norm{D_1\mathcal{E}(\lambda,u)[\zeta]}_{L^{'3}_{-3}}&\leq&\frac{\abs{\zeta}}{6}\norm{D_1\tilde{R}_{\lambda}(\Phi_0+u)}_{L^{'3}_{-3}}+\frac{\abs{\zeta}}{6}\norm{\frac{D_1\tilde{K}^{2}_{\lambda}}{(\Phi_0+u)^5}}_{L^{'3}_{-3}}\nonumber\\
&\leq& C \abs{\zeta}
\end{eqnarray}
similarly, by definition of $O_y$ and lemma (\ref{lem5}) we have
\begin{eqnarray}
\norm{D_2\mathcal{E}(\lambda,u)[v]}_{L^{'3}_{-3}}&\leq&\norm{\Delta_4 v}_{L^{'3}_{-3}}+\frac{1}{6}\norm{\tilde{R}_{\lambda}v}_{L^{'3}_{-3}}
+\norm{\frac{5\tilde{K}^{2}_{\lambda}}{6(\Phi_0+u)^6}v}_{L^{'3}_{-3}}\nonumber\\
&\leq&C\norm{v}_{W^{'2,3}_{-1}}.
\end{eqnarray}
2. To show $D_1\mathcal{E}(\lambda,u)$ is continuous (it is in fact uniformly continuous), we use the triangle inequality, inequality (\ref{38}), and lemma (\ref{lem5}). Then  
\begin{eqnarray}
\norm{D_1\mathcal{E}(\lambda_1,u)-D_1\mathcal{E}(\lambda_2,u)}_{L^{'3}_{-3}}&\leq&\frac{1}{6}\norm{(D_1\tilde{R}_{\lambda_1}-D_1\tilde{R}_{\lambda_2})(\Phi_0+u)}_{L^{'3}_{-3}}\nonumber\\
&+&\norm{\frac{D_1\tilde{K}^{2}_{\lambda_1}-D_1\tilde{K}^{2}_{\lambda_2}}{6(\Phi_0+u)^5}}_{L^{'3}_{-3}}\nonumber\\
&\leq& C\abs{\lambda_1-\lambda_2}.
\end{eqnarray}
To prove continuity in $u$ consider the following identity for arbitrary $x$, $y$ and integer $p$:
\begin{equation}
\frac{1}{x^p}-\frac{1}{y^p}=(y-x)\sum_{i=0}^{p-1}x^{i-p}y^{-1-i}.\label{40}
\end{equation}
Then
\begin{equation}
r^{-7}\left(\frac{1}{\left(\Phi_0+u_1\right)^6}-\frac{1}{\left(\Phi_0+u_2\right)^6}\right)=\left(u_2-u_1\right)M.
\end{equation}
where
\begin{equation}
M=\sum_{i=0}^{5}\left(r\left(u+\Phi_0\right)\right)^{i-6}\left(r\Phi_0\right)^{-1-i}.
\end{equation}
Since $u_1,u_2\in O_y$, and using the lower bound in equation (\ref{38}) we have
\begin{equation}
M\leq C.\label{45}
\end{equation}
Then by (\ref{45}) and Lemma \ref{lem3}-\ref{44} we have 
\begin{eqnarray}
\norm{D_2\mathcal{E}(\lambda,u_1)[v]-D_2\mathcal{E}(\lambda,u_2)[v]}_{L^{'3}_{-3}}&=&\norm{v\frac{5\tilde{K}^{2}_{\lambda}}{6(\Phi_0+u_1)^6}-v\frac{5\tilde{K}^{2}_{\lambda}}{6(\Phi_0+u_2)^6}}_{L^{'3}_{-3}}\nonumber\\
&\leq&C\norm{\frac{\left(u_1-u_2\right)v}{r}}_{L^{'3}_{-3}}
\end{eqnarray}
The right hand side of the above equation can be bounded as follows: (we write $\td x$ to represent the volume element for the Euclidean metric on $\mathbb{R}^4 \backslash \{ 0 \} $)
\begin{eqnarray}
\norm{\frac{\left(u_1-u_2\right)v}{r}}_{L^{'3}_{-3}}&=&\left(\int_{\mathbb{R}^4 \backslash \{0\}}\frac{\left(u_1-u_2\right)^3v^3}{r^3}r^5 dx\right)^{1/3}\nonumber\\
&=&\left(\int_{\mathbb{R}^4 \backslash \{0\}}\frac{\left(u_1-u_2\right)^3\left(rv\right)^3}{r} \td x\right)^{1/3}\nonumber\\
&\leq& C\norm{v}_{W^{'2,3}_{-1}}\left(\int_{\mathbb{R}^4 \backslash \{0\}}\frac{\left(u_1-u_2\right)^3}{r} \td x\right)^{1/3} \nonumber\\
&\leq& C\norm{v}_{W^{'2,3}_{-1}}\norm{u_1-u_2}_{W^{'2,3}_{-1}}.\label{in2}
\end{eqnarray}
The first inequality follows from Lemma \ref{lem4} and the second inequality from the definition of Sobolev norms. Therefore, we have 
\begin{equation}
\norm{D_2\mathcal{E}(\lambda,u_1)[v]-D_2G(\lambda,u_2)[v]}_{L^{'3}_{-3}}\leq  C\norm{v}_{W^{'2,3}_{-1}}\norm{u_1-u_2}_{W^{'2,3}_{-1}}.
\end{equation}
Thus, $D_2G(\lambda,u)$ is a continuous operator.\\
3.  Equation (\ref{30}) is straightforward to prove. We prove (\ref{31}) as follows
\begin{eqnarray}
\mathcal{E}(\lambda,u+v)-\mathcal{E}(\lambda,u)-D_2\mathcal{E}(\lambda,u)[v]&=& \frac{\tilde{K}^{2}_{\lambda}}{6}\left(\frac{1}{\left(\Phi_0+u+v\right)^5}-\frac{1}{\left(\Phi_0+u\right)^5}+\frac{5v}{\left(\Phi_0+u\right)^6}\right)\nonumber
\end{eqnarray}
By simplifying we have
\begin{equation}
r^{-7}\left(\frac{1}{\left(\Phi_0+u+v\right)^5}-\frac{1}{\left(\Phi_0+u\right)^5}+\frac{5v}{\left(\Phi_0+u\right)^6}\right)=v^2 M_1.
\end{equation}
where 
\begin{equation}
M_1=\frac{1}{\left(r\left(\Phi_0+u+v\right)\right)^5\left(r\left(\Phi_0+u\right)\right)^6}\sum_{\substack{{i+j+k=4}\\
\forall i,j,k\geq 0}}C_{ijk}\left(r\Phi_0\right)^{i}\left(ru\right)^{j}\left(rv\right)^{k}.
\end{equation}
Where $C_{ijk}$ are numerical constants. To find the bound of $M_1$ we will use equation (\ref{42}) and the fact that $u,v\in V$. Then we have
\begin{equation}
\abs{M_1}\leq C \frac{(r^2+ab+b^2)(r^2+ab+a^2)+\mu^2}{\left(\left[(r^2+ab+b^2)(r^2+ab+a^2)+\mu^2\right]^{1/4}-C_0\xi\right)^{11}}\leq C.\label{11}
\end{equation}
Then by lemma \ref{lem5} and above inequality we have
\begin{eqnarray}
\norm{\mathcal{E}(\lambda,u+v)-\mathcal{E}(\lambda,u)-D_2\mathcal{E}(\lambda,u)[v]}_{L^{'3}_{-3}}&\leq& C\norm{\frac{v^2M_1}{r}}_{L^{'3}_{-3}}\leq C\norm{v}^2_{W^{'2,3}_{-1}}.
\end{eqnarray}
By steps similar to \eqref{in2} we have second inequality. Hence, we have proved statements (1),(2), and (3) and $\mathcal{E}(\lambda,u):\mathbb{R}\times O_y\rightarrow L^{'3}_{-3}$ is $C^1$.
\subsubsection{$D_2\mathcal{E}(0,0)$ is an isomorphism}
We now verify $D_2\mathcal{E}(0,0):W^{'2,3}_{-1}\rightarrow L^{'3}_{-3}$ is an isomorphism. We write this linear operator as
\begin{equation}
D_2\mathcal{E}(0,0)[v]=\Delta_4 v-\alpha v\label{43}
\end{equation}
where
\begin{equation}
\alpha=\frac{\tilde{R}_{0}}{6}+\frac{5\tilde{K}^{2}_{0}}{6\Phi_0^6}\label{alpha}
\end{equation}
An important property of the function $\alpha$ by lemma \ref{alphalem} is a nonnegative bounded function in $\mathbb{R}^4 \backslash \{0\}$, that is $\alpha=hr^{-6}$ where $h\geq 0$.  Therefore $\alpha \in L'^3_{-3}$.  Hence, as shown in Appendix A, when $M=\mathbb{R}^4 \backslash\{0\}$ and $p=3, \delta = -1$, the map $\Delta_{4}-\alpha$ is an isomorphism from $W^{'2,3}_{-1}\rightarrow L^{'3}_{-3}$.

\section{Discussion}
We have constructed a one-parameter family of initial data to the vacuum Einstein's equations with the same symmetries and asymptotic behaviour as initial data for the extreme Myers-Perry black hole in five dimensions. In particular this data have the same angular momenta $(J_1,J_2)$.  Such initial data will generically have a non-stationary evolution and is a starting point to investigate the dynamics near extremality for such black holes.  Our results generalize the analogous results concerning initial data `close' to extreme Kerr data \cite{Dain:2010uh}.  An important property of this three-dimensional initial data is that they had strictly greater energy than extreme Kerr. This is a  consequence of Dain's mass-angular momentum inequality, valid in axisymmetry: $M \geq \sqrt{J}$, for which the initial data for extreme Kerr is the unique minimizer that saturates the bound \cite{dain2006proof,dain2008proof,dain2012geometric}.  In our case, however,  in the absence of geometric inequalities we cannot conclude that the energy of the family of initial data discussed is strictly greater than that of extreme Myers-Perry.  Noting that the mass of Myers-Perry black holes satisfy the bound
\begin{equation}
M^3 \geq \frac{27\pi}{32}\left(|J_1| + |J_2|\right)^2
\end{equation} with equality in the extreme case, it would be tempting to conjecture Dain's inequality admits a generalization to four-dimensional biaxisymmetric initial data.  Proving that the extreme Myers-Perry initial data is local minimizer of energy amongst the class of initial data we have considered here would be a useful first step towards establishing an analogue of Dain's global result.  Note, however, that the energy of an extreme black ring \cite{Pomeransky:2006bd} satisfies 
\begin{equation}
M^3 = \frac{27\pi}{4} |J_1| (|J_2| - |J_1|)
\end{equation} This suggests a more complicated geometric inequality in five dimensions, which takes into account which combination of rotational Killing fields have fixed points in the interior of the initial data.   We hope to address these issues in future work. 

The method used here to find solutions of the constraint equations relied on the ability to generate initial data from the specification of scalar functions and reduce the problem to a single scalar PDE. In particular, the assumption of `$t-\phi^i$' symmetry allows one to determine the extrinsic curvature completely from the twist potentials. The existence of these potentials in turn relied on the existence of $U(1)^2$ isometry. It is clear that this technique would work in spacetime dimensions $D>5$, provided one assumes $U(1)^{D-3}$ isometry. Of course, this is too much Abelian symmetry to describe an asymptotically flat black hole for $D>5$. However, in certain limits extra Abelian symmetry may arise. For example, for higher-dimensional  Myers-Perry black holes, one may take an `ultraspinning limit' which enhances the number of commuting isometries (the limit changes the black hole horizon from $S^{D-2}$ to  $S^p \times \mathbb{T}^q$ for appropriately chosen integers $(p,q)$) \cite{Figueras:2008qh}.  It is known that \emph{non-extremal} black holes with a single non-zero angular momentum admit a linearized gravitational instability in the ultraspinning limit \cite{Dias:2010maa}. It might be interesting to investigate ultraspinning instabilities of \emph{extremal} black holes in $D>5$ using the formalism described here.  The initial data under consideration would have, in addition to a cylindrical end, an asymptotically Kaluza-Klein end, rather than an asymptotically flat one. 

\section*{Acknowledgements} 
AA is supported by a graduate scholarship from Memorial University. HKK is supported by an NSERC Discovery Grant. We would like to thank Ivan Booth and Chris Radford for useful comments and discussions. We also especially thank Sergio Dain and Eugenia Gabach-Clement for reading a draft of the manuscript and for a number of helpful suggestions and comments. This research was supported in part by Perimeter Institute for Theoretical Physics. Research at Perimeter Institute is supported by the Government of Canada through Industry Canada and by the Province of Ontario through the Ministry of Economic Development and Innovation. 
 
\appendix
\section{Asymptotically Euclidean manifolds}\label{appA}
A precise mathematical formalism to describe the asymptotic behaviour of functions on a space is the theory of weighted Sobolev spaces.  Here we use Bartnik's weighted Sobolev space \cite{bartnik1986mass,maxwell2005solutions} which is appropriate for Riemannian manifolds with asymptotically Euclidean and cylindrical ends. The weight function is $r=\abs{x}$ for $x\in \mathbb{R}^n$. Then for any $\delta\in \mathbb{R}$, $1\leq p<\infty$,  Bartnik's weighted Sobolev space $W^{'k,p}_{\delta}$ is the subset of $W^{'k,p}_{\text{loc}}$ for which the norm 
\begin{equation}
\norm{u}_{W^{'k,p}_{\delta}}=\sum_{j=0}^k\norm{\partial^ju}_{L^{'p}_{\delta-j}},\quad\text{where}\quad\norm{u}_{L^{'p}_{\delta}}=\Bigg(\int_{\mathbb{R}^n\backslash \{0\}}\abs{u}^pr^{-\delta p-n} \td x \Bigg)^{1/p}
\end{equation}
is finite.  Relevant properties of this weighted Sobolev space are summarized in the following lemma \cite{bartnik1986mass,maxwell2005solutions,Dain:2010uh}
\begin{lemma}\label{lem4} $\phantom{4}$:
\begin{enumerate}
\item If $p\leq q$ and $\delta_1<\delta_2$ then $L^{'p}_{\delta_1}\subset L^{'q}_{\delta_2}$ and the inclusion is continuous.
\item For $k\geq 1$ and $\delta_1<\delta_2$ the inclusion $W^{'k,p}_{\delta_1}\subset W^{'k-1,p}_{\delta_2}$ is compact.
\item If $1/p<k/n$ then $W^{'k,p}_{\delta}\subset C^{'0}_{\delta}$. The inclusion is continuous. That is if $u \in W^{' k,p}_{\delta}$ then $r^{-\delta}\abs{u} \leq C \norm{u}_{W^{'k,p}_{\delta}}$. Further, as proved in \cite{Dain:2010uh}, $\lim_{r\to 0} r^{-\delta} \abs{u} = \lim _{r \to \infty} r^{-\delta} \abs{u} =0$. 
\end{enumerate}
\end{lemma}
\noindent Let $M$ be a smooth, connected, complete, $n$-dimensional Riemannian manifold $(M,\gamma)$, and let $\rho<0$. We say $(M,\gamma)$ is asymptotically Euclidean of class $W^{'k,p}_{\rho}$ if
\begin{itemize}
\item The metric $\gamma\in W^{'k,p}_{\rho}(M)$, where $1/p-k/n<0$ and $\gamma$ is continuous.
\item There exists a finite collection $\{N_i\}_{i=1}^m$ of open subsets of $M$ and diffeomorphisms
$\Phi_i:E_r\rightarrow N_i$ ($E_r=\mathbb{R}^n \backslash  \bar{B}_r(0)$) such that $M-\cup_iN_i$ is compact.
\item For each $i$, $\Phi^*_i\gamma-\bar{\gamma}\in W^{'k,p}_{\rho}(E_r)$
\end{itemize} 
We call the charts $\Phi_i$ end charts and the corresponding coordinates are end coordinates. Now, suppose $(M,\gamma)$ is asymptotically Euclidean, and let $\{\Phi_i\}_{i=1}^{m}$ be its collection of end charts. Let $K=M-\cup_i\Phi_i(E_{2r})$, so $K$ is a compact manifold. The weighted Sobolev space $W^{k,p}_{\delta}(M)$ is the subset of $W^{k,p}_{\text{loc}}(M)$ such that the norm 
\begin{equation}
\norm{u}_{W^{k,p}_{\delta}(M)}=\norm{u}_{W^{k,p}(K)}+\sum_i\norm{\Phi_i^*u}_{W^{k,p}_{\delta}(E_{r})}
\end{equation}
is finite. We can define similarly weighted Lebesgue space $L^{'p}_{\delta}(M)$ and $C^{'k}_{\delta}$ and $C^{'\infty}_{\delta}(M)=\cap_{k=0}^{\infty}C^{'k}_{\delta}(M)$. In the  particular case when $M=\mathbb{R}^n$, then we have just one asymptotically Euclidean end. Moreover, if $(M,\gamma)$ is an asymptotically Euclidean manifold of class $W^{'k,p}_{\rho}$, we say $(M,\gamma,K)$ is asymptotically Euclidean dataset if $K\in W^{'k-1,p}_{\rho-1}(M)$.\\

\noindent The main goal of this appendix is to consider the Poisson operator $\mathcal{L}=\Delta_{\gamma}-\alpha$ on scalar functions of an asymptotically Euclidean manifold and express a very classical result (\cite{mcowen1979behavior} or see \cite{maxwell2005solutions}), that is, $\mathcal{L}$ is an isomorphism from Sobolev space $W^{'2,p}_{\delta}$ to $L^{'p}_{\delta}$. We start with the estimate \cite{choquet1981elliptic,maxwell2005solutions,choquet2009general} 
\begin{lemma}\label{lem2}
Suppose $(M,\gamma)$ is asymptotically Euclidean of class $W^{'2,p}_{\rho}$, $p>\frac{n}{2}$, $\rho<0$. Then if $2-n<\delta<0$, $\delta'\in \mathbb{R}$, and $u\in W^{'2,p}_{\delta}$ we have
\begin{equation}
\norm{u}_{W^{'2,p}_{\delta}}\leq \norm{\mathcal{L}u}_{L^{'p}_{\delta-2}}+\norm{u}_{L^{'p}_{\delta'}}.
\end{equation}
\end{lemma}
\noindent Now we have following weak maximum principle (Lemma 3.2 in \cite{maxwell2005solutions})
\begin{lemma}\label{lem1}
Suppose $(M,\gamma)$ is asymptotically Euclidean of class $W^{'k,p}_{\rho}$, $k\geq 2$, $k>\frac{n}{p}$, and suppose $\alpha\in W^{'k-2,p}_{\rho-2}$ and suppose $\alpha\geq 0$. If $u\in W^{'k,p}_{\text{loc}}$ satisfies
\begin{equation}
-\Delta_{\gamma} u+\alpha u\leq 0\label{55}
\end{equation}
and if $u^+\equiv \text{max}(u,0)$ is $o(1)$ on each end of $M$, then $u\leq 0$. In particular, if $u\in  W^{'k,p}_{\delta}$ for some $\delta<0$ and $u$ satisfies (\ref{55}), then $u\leq 0$.
\end{lemma}
\begin{proof}
Fix $\epsilon>0$, and let $v=(u-\epsilon)^+$. Since $u^+=o(1)$ on each end, we see $v$ is compactly supported. Moreover, since $u\in W^{'k,p}_{\text{loc}}$ we have from Sobolev embedding that $u\in W^{'1,2}_{\text{loc}}$ and hence $v\in W^{'1,2}$. Now,  
\begin{equation}
\int_M \left(-v\Delta_{\gamma} u+\alpha uv\right)\, \td x\leq 0\quad\Longrightarrow \quad\int_M -v\Delta_{\gamma} u\, \td x \leq -\int_M\alpha uv \, \td x\leq 0
\end{equation} where $\td x$ denotes the volume element on $(M,\gamma)$. 
Since $\alpha\geq 0$, $v\geq 0$ and u is positive wherever $v\neq 0$. Integrating by parts we have
\begin{equation}
\int_M\abs{\nabla v}^2\, dx\leq 0
\end{equation}
since $\nabla u=\nabla v$ on the support of $v$. So $v$ is constant and compactly supported, so it should be zero, i.e. $\text{max}(u-\epsilon,0)=0$. Then we conclude $u\leq \epsilon$. Sending $\epsilon$ to $0$ we have $u\leq 0$.\\
Now, if $u\in W^{'k,p}_{\delta}$, since $W^{'k,p}_{\delta}\subset C^{'0}_{\delta}$,, we have $u\in  C^{'0}_{\delta}$. Hence if $\delta<0$, then $u^+=o(1)$ and lemma can be applied to $u$. 
\end{proof}
Using this Lemma we can prove the following theorem.
\begin{thm}
Suppose $(M,\gamma)$ is asymptotically Euclidean of class $W^{'2,p}_{\rho}$, $p>\frac{n}{2}$. Then if $2-n<\delta<0$  and $\alpha \in L'^p_{\delta-2}$, the operator $\mathcal{L}:W^{'2,p}_{\delta}\rightarrow L^{'p}_{\delta-2}$  is Fredholm with index $0$. Moreover, if $\alpha\geq 0$ then $\mathcal{L}$ is an isomorphism.
\end{thm}
\begin{proof}
By the estimate in Lemma \ref{lem2} and \cite{choquet1981elliptic} this operator is Fredholm. Now we show $\mathcal{L}$ is injective. Let $\mathcal{L}u=0$ for $u\in W^{'2,p}_{\delta}$. Then by weak maximum principle we have $u=0$ on $M$ for $2-n<\delta<0$ and  $\mathcal{L}$ is injective. To show $\mathcal{L}$ is surjective, it suffices to show $\mathcal{L}^*$ is injective from $L^{'p}_{2-n-\delta}\rightarrow W^{'-2,p}_{-n-\delta}$. Now let $f_1$ and $f_2$ be smooth and compactly supported in each end of $M$. We have from integration by parts
\begin{eqnarray}
0=\left<f_2,\mathcal{L}^*(f_1)\right>=\left<\mathcal{L}(f_2),f_1\right>=\int_M \mathcal{L}(f_2)f_1\, \td x
\end{eqnarray}
Thus $\int_M \mathcal{L}(f_2)f_1\, \td x=0$ for all smooth and compactly supported $f_2$ in each end of $M$, then $f_1=0$ and $\mathcal{L}^*$ is injective. Then $\mathcal{L}$ is surjective. Therefore, $\mathcal{L}$ is an isomorphism.
\end{proof}

\section{Myers-Perry black hole initial data}\label{AppC}
In this Appendix we will give details on various properties of the initial data for the extreme Myers-Perry metric. We have used MAPLE  to simplify a number of our computations.  Our main interest is to find certain final bounds and since most of the calculations are similar, we only provide explicit details for a subset of cases. The slice metric can be written as
\begin{eqnarray}
h=\frac{\Sigma}{r^2}\left(dr^2
+ r^2 \td\theta^2\right)+\sigma_{ij}\td\phi^i\td\phi^j
\end{eqnarray}
where 
\begin{eqnarray}
\Sigma&=&r^2+ab+a^2\cos^2\theta+b^2\sin^2\theta\qquad\quad \sigma_{12}=\frac{ab\mu}{\Sigma}\sin^2\theta\cos^2\theta\\
\sigma_{11}&=&\frac{a^2\mu}{\Sigma}\sin^4\theta+(r^2+ab+a^2)\sin^2\theta\qquad 
\sigma_{22}=\frac{b^2\mu}{\Sigma}\cos^4\theta+(r^2+ab+b^2)\cos^2\theta\nonumber
\end{eqnarray}
where $\phi^1=\varphi$ and $\phi^2=\psi$. Now if we choose $\rho=\frac{1}{2}r^2\sin 2\theta$ and $z=\frac{1}{2}r^2\cos 2\theta$, then the conformal slice metric of  the extreme Myers-Perry black hole can be written 
\begin{eqnarray}
\tilde{h}=e^{2U}\left(\td\rho^2 + \td z^2\right)+\sigma'_{ij}\td\phi^i\td\phi^j
\end{eqnarray}
where
\begin{equation}\label{phi}
\Phi_0^2=\frac{\sqrt{\det\sigma}}{\rho}\qquad
\sigma'_{ij}=\Phi_0^{-2}\sigma_{ij}\qquad 
e^{2U}=\Phi_0^{-2}\frac{\Sigma}{r^4}
\end{equation}
In general, the lapse and shift vectors are degrees of freedom for the initial data set. But  since we want to preserve geometrical properties of the initial data under evolution, we compute the lapse of the extreme Myers-Perry spacetime and select the shift vector  to be the product of $r$ and the shift of extreme Myers-Perry metric. 
{\small
\begin{eqnarray}
\alpha&=&\sqrt{\frac{r^4\Sigma}{(\Sigma+\mu)r^4+\mu^2\left(r^2+ab\right)}},\label{laps}\\
\beta^{\varphi}&=&\frac{ra\mu(r^2+ab+b^2)}{(\Sigma+\mu)r^4+\mu^2\left(r^2+ab\right)},\quad\quad\quad\quad
\beta^{\psi}=\frac{rb\mu(r^2+ab+a^2)}{(\Sigma+\mu)r^4+\mu^2\left(r^2+ab\right)}\label{shift}
\end{eqnarray}}
In addition, we showed in section \ref{sec2} that the extrinsic curvature can be generated from scalar potentials $\omega_{\phi^i}$. In the coordinate system used above, these are
{\small
\begin{eqnarray}
\omega_{\varphi}=\frac{a(a^2-b^2)(r^2+ab+b^2)\cos^2\theta-r^2a(2a^2+2ab+r^2)}{(a-b)^2}+\frac{a(r^2+ab+a^2)^2(r^2+ab+b^2)}{\Sigma(a-b)^2}\nonumber\\
\end{eqnarray}
\begin{eqnarray}
\omega_{\psi}=\frac{br^2((a+b)^2+r^2)-b(a^2-b^2)(r^2+ab+a^2)\cos^2\theta}{(a-b)^2}-\frac{b(r^2+ab+a^2)(r^2+ab+b^2)^2}{\Sigma(a-b)^2}\nonumber\\
\end{eqnarray}}
It is important to mention
\begin{equation}
\Delta_2=\frac{\partial^2}{\partial\rho^2}+\frac{\partial^2}{\partial z^2}=\frac{1}{r^4}\left(r^2\frac{\partial^2}{\partial r^2}+r\frac{\partial}{\partial r}+\frac{\partial^2}{\partial\theta^2}\right)
\end{equation}
Now we will prove some useful lemmas for the main theorem.
\begin{lemma}\label{alphalem}
The function $\alpha$ in equation \eqref{alpha}  is nonnegative and has following bounds
\begin{equation}
\alpha=\frac{\tilde{K}^{2}_{0}}{2\Phi_0^6}+r^2(dv)^2=hr^{-6}
\end{equation}
where $h$ is a bounded nonnegative function.
\end{lemma}
\begin{proof}
First we know by conformal transformation $h_{ab}=\Phi^2\tilde{h}_{ab}$ the scalar curvature will be\footnote{There are some typos about factors in journal version.}
\begin{equation}
\tilde{R}=R\Phi^2+6\left(\Delta_{\tilde{h}}v+|\tilde{\nabla} v|^2\right)
\end{equation}
where $v=\log\Phi$. By constraint equations \eqref{vaccons1} and the fact that conformal extreme Myers-Perry satisfies in relation
\begin{equation}
\Delta_{\tilde{h}}v=-\frac{1}{2\Phi^6}\tilde{K}_{ab}\tilde{K}^{ab}
\end{equation}
we have
\begin{eqnarray}
\tilde{R}&=&{K}_{ab}{K}^{ab}\Phi^2-3\Phi^{-6}\tilde{K}_{ab}\tilde{K}^{ab}+6|\tilde{\nabla} v|^2\nonumber\\
&=&-2\tilde{K}_{ab}\tilde{K}^{ab}\Phi^{-6}+6e^{-2U}(dv)^2
\end{eqnarray}
Then by equations \eqref{curv1} and \eqref{K1} we have
\begin{equation}
\tilde{R}_0=-5\tilde{K}^2_0\Phi^{-6}+6r^2(dv)^2
\end{equation}
Therefore, $\alpha$ is
\begin{equation}
\alpha=\frac{\tilde{R}_{0}}{6}+\frac{5\tilde{K}^{2}_{0}}{6\Phi_0^6}=\frac{\tilde{K}^{2}_{0}}{2\Phi_0^6}+r^2(dv)^2=hr^{-6}
\end{equation}
\end{proof}
\begin{lemma}\label{lem3}
Let $\Phi_0$, $\tilde{R}_0$, and $\tilde{K}^2_0$ be defined as in (\ref{phi}), (\ref{curv1}), and (\ref{33}), respectively. Then we have following bounds:
\begin{enumerate}
\item $(ab\mu)^{1/4}\leq \left[(r^2+ab+b^2)(r^2+ab+a^2)\right]^{1/4}\leq r\Phi_0 \leq\left[(r^2+ab+b^2)(r^2+ab+a^2)+\mu^2\right]^{1/4}\label{21}$
\item $\abs{\tilde{R}_0}\leq \frac{C}{r^4}$ and $\abs{\tilde{K}^2_0}\leq\frac{C}{r^6}$\label{44}
\item  $|\Delta_4\Phi_0|\leq\frac{C}{r^6}$\label{10}
\end{enumerate}
\end{lemma}
\begin{proof} We will prove just 1 here;  the remaining bounds require lengthy algebraic manipulations.
\begin{enumerate}
\item We have 
\begin{eqnarray}
r^2\Phi_0^2&=&\left[(r^2+ab+b^2)(r^2+ab+a^2)+\frac{\mu(r^2+ab)(a^2\cos^2\theta+b^2\sin^2\theta)+\mu a^2b^2}{\Sigma}\right]^{1/2}\nonumber\\
&\leq&\left[(r^2+ab+b^2)(r^2+ab+a^2)+\mu^2\right]^{1/2}\nonumber\\
\end{eqnarray}
so if $r\rightarrow\infty$ then we have minimum of $r^2\Phi_0^2$
\begin{equation}
\sqrt{(r^2+ab+b^2)(r^2+ab+a^2)}\leq r^2\Phi_0^2 
\end{equation}
Therefore for $a,b>0$  we have 
{\small\begin{equation}
  (ab\mu)^{1/4}\leq \left[(r^2+ab+b^2)(r^2+ab+a^2)\right]^{1/4}\leq r\Phi_0 \leq\left[(r^2+ab+b^2)(r^2+ab+a^2)+\mu^2\right]^{1/4}
\end{equation}}
\end{enumerate}
\end{proof}
\begin{lemma}\label{lem5}
If we transform metric functions by (\ref{34}) for small $\lambda$ (i.e.$ -\lambda_0 < \lambda < \lambda_0$)  then
\begin{enumerate}
\begin{minipage}{0.4\linewidth} 
\item $\norm{\tilde{R}_{\lambda}}_{L'^3_{-3}}\leq C$
\item $\norm{\tilde{K}^{2}_{\lambda}}_{L'^3_{-3}}\leq C$
\item $\norm{D_1\tilde{R}_{\lambda}}_{L'^3_{-3}}\leq C$
\end{minipage} 
\begin{minipage}{0.6\linewidth}
\item $\norm{D_1\tilde{K}^{2}_{\lambda}}_{L'^3_{-3}}\leq C$
\item $\norm{D_1\tilde{R}_{\lambda_1}-D_1\tilde{R}_{\lambda_2}}_{L'^3_{-3}}\leq C\abs{\lambda_1-\lambda_2}$
\item $\norm{D_1\tilde{K}^{2}_{\lambda_1}-D_1\tilde{K}^{2}_{\lambda_2}}_{L'^3_{-3}}\leq C\abs{\lambda_1-\lambda_2}$
\end{minipage} 
\end{enumerate}
\end{lemma}
\begin{proof} We will prove numbers 1 and 4 of these inequalities and others will be similar.
1) By definition of $\tilde{R}_{\lambda}$   we have
{\small 
\begin{eqnarray}
\tilde{R}_{\lambda}&=&-r^2\Delta_2(U+\lambda U_1)+r^2\frac{\det(\td\sigma'+\lambda\td\bar{\sigma})}{2\rho^2}\nonumber\\
&=&-r^2\Delta_2 U-r^2\lambda\Delta_2 U_1+\frac{r^2}{2\rho^2}\left[(\td\sigma'_{11}+\lambda\td\bar{\sigma}_{11})\cdot (\td\sigma'_{22}+\lambda\td\bar{\sigma}_{22})-(\td\sigma'_{12}+\lambda\td\bar{\sigma}_{12})^2\right]\nonumber\\
&=&\tilde{R}_0-r^2\lambda\Delta_2 U_1+\frac{r^2}{2\rho^2}\left[\lambda\td\bar{\sigma}_{11}\cdot \td\sigma'_{22}+\lambda(\td\sigma'_{11}+\lambda\td\bar{\sigma}_{11})\cdot \td\bar{\sigma}_{22}-\lambda(2\td\sigma'_{12}+\lambda\td\bar{\sigma}_{12})\cdot\td\bar{\sigma}_{12}\right]\nonumber\\
\end{eqnarray}}
Then by triangle inequality we have
{\small
\begin{eqnarray}
\norm{\tilde{R}_{\lambda}}_{L^{'3}_{-3}}&\leq&\norm{\tilde{R}_0}_{L^{'3}_{-3}}+|\lambda|\norm{r^2\Delta_2 U_1}_{L^{'3}_{-3}}+\norm{\frac{r^2}{2\rho^2}\left(\lambda\td\bar{\sigma}_{11}\cdot \td\sigma'_{22}+\lambda(\td\sigma'_{11}+\lambda\td\bar{\sigma}_{11})\cdot \td\bar{\sigma}_{22}\right)}_{L^{'3}_{-3}}\nonumber\\
&+&\norm{\frac{r^2}{2\rho^2}\left(\lambda(2\td\sigma'_{12}+\lambda\td\bar{\sigma}_{12})\cdot\td\bar{\sigma}_{12}\right)}_{L^{'3}_{-3}}\nonumber\\
&\leq&C
\end{eqnarray}}
We used inequality of Lemma \ref{lem3}-\ref{44} and the fact that functions ${U}_1$ and $\bar{\sigma}_{ij}$  have compact support outside the origin.\\

\noindent 4) By definition of full contraction of extrinsic curvature we have
\begin{eqnarray}
\tilde{K}^2_{\lambda}&=&\frac{r^2}{2\rho^4}\Bigg[(\sigma'_{11}+\lambda\bar{\sigma}_{11})(\td\omega_{\psi}+\lambda\td\omega_{2})^2+(\sigma'_{22}+\lambda\bar{\sigma}_{22})(\td\omega_{\varphi}+\lambda\td\omega_{1})^2\nonumber\\
&-&2(\sigma'_{12}+\lambda\bar{\sigma}_{12})(\td\omega_{\psi}+\lambda\td\omega_{2})\cdot(\td\omega_{\varphi}+\lambda\td\omega_{1})\Bigg]
\end{eqnarray}
Then have
\begin{eqnarray}
D_1\tilde{K}^2_{\lambda}&=&\frac{r^2}{2\rho^4}\Bigg[\bar{\sigma}_{11}(\td\omega_{\psi}+\lambda\td\omega_{2})^2+2(\sigma'_{11}+\lambda\bar{\sigma}_{11})\td\omega_{2}\cdot(\td\omega_{\psi}+\lambda\td\omega_{2})+\bar{\sigma}_{22}(\td\omega_{\varphi}+\lambda\td\omega_{1})^2\nonumber\\
&+&(\sigma'_{22}+\lambda\bar{\sigma}_{22})\td\omega_{1}\cdot(\td\omega_{\varphi}+\lambda\td\omega_{1})-2\lambda\bar{\sigma}_{12}(\td\omega_{\psi}+\lambda\td\omega_{2})\cdot(\td\omega_{\varphi}+\lambda\td\omega_{1})\nonumber\\
&-&4(\sigma'_{12}+\lambda\bar{\sigma}_{12})\lambda\td\omega_{2}\cdot(\td\omega_{\varphi}+\lambda\td\omega_{1})-4(\sigma'_{12}+\lambda\bar{\sigma}_{12})(\td\omega_{\psi}+\lambda\td\omega_{2})\cdot\td\omega_{1}\Bigg]
\end{eqnarray}
Then by triangular inequality and the fact that $\omega_i$ has compact support outside axis and $\bar{\sigma}_{ij}$ has compact support outside the origin one can show it is bounded.
\end{proof}
\noindent Finally,we have following limits
\begin{eqnarray}
\lim_{s\rightarrow\infty} \omega_{\varphi}&=&\frac{ab\mu\cos^2\theta}{(a-b)}+\frac{a^3b\mu(a+b)}{(ab+a^2\cos^2\theta+b^2\sin^2\theta)(a-b)^2}\\
\lim_{s\rightarrow\infty} \omega_{\psi}&=&-\frac{ab\mu\cos^2\theta}{(a-b)}-\frac{ab\mu(a+b)}{(ab+a^2\cos^2\theta+b^2\sin^2\theta)(a-b)^2}\\
\lim_{s\rightarrow\infty} r\Phi_0&=&\frac{\mu(\mu-ab)}{2(ab+a^2\cos^2\theta+b^2\sin^2\theta)} \; . 
\end{eqnarray}

\bibliographystyle{unsrt}
\bibliographystyle{abbrv}  
\bibliography{masterfile}
       
\end{document}